%% file: main.tex
\documentclass[aps,physrev,twocolumn,superscriptaddress]{revtex4-2}
\input{pretex}

\usepackage{amsmath,amsfonts,amssymb,array,graphicx,mathtools,multirow,bm,times,tcolorbox,relsize,booktabs}
\usepackage[utf8]{inputenc}
\usepackage[T1]{fontenc}
\usepackage{qcircuit}
\usepackage{pifont}
\usepackage{svg}

\begin{document}

\title{No-Go Theorems for Universal Quantum State Purification via Classically Simulable Operations}

\author{Keming He}
\author{Chengkai Zhu}
\author{Hongshun Yao}
\affiliation{Thrust of Artificial Intelligence, Information Hub, \\ The Hong Kong University of Science and Technology (Guangzhou), Guangdong 511453, China.}

\author{Jinguo Liu}
\affiliation{Thrust of Advanced Materials, Function Hub, \\ The Hong Kong University of Science and Technology (Guangzhou), Guangdong 511453, China.}

\author{Yinan Li}
 \email{yinan.li@whu.edu.cn}
\affiliation{School of Artificial Intelligence, Wuhan University, Hubei 430072, China}
\affiliation{National Center for Applied Mathematics in Hubei, Hubei 430072, China}
\affiliation{Hubei Key Laboratory of Computational Science, Hubei 430072, China}

\author{Xin Wang}
\email{felixxinwang@hkust-gz.edu.cn}
\affiliation{Thrust of Artificial Intelligence, Information Hub, \\ The Hong Kong University of Science and Technology (Guangzhou), Guangdong 511453, China.}



\date{\today}

\begin{abstract}
Quantum state purification, a process that aims to recover a state closer to a system's principal eigenstate from multiple copies of an unknown noisy quantum state, is crucial for restoring noisy states to a more useful form in quantum information processing. Fault-tolerant quantum computation relies on stabilizer operations, which are classically simulable protocols critical for error correction but inherently limited in computational power. In this work, we investigate the limitations of classically simulable operations for quantum state purification. We demonstrate that while certain classically simulable operations can enhance fidelity for specific noisy state ensembles, they cannot achieve universal purification. We prove that neither deterministic nor probabilistic protocols using only classically simulable operations can achieve universal purification of two-copy noisy states for qubit systems and all odd dimensions. We further extend this no-go result of state purification using three and four copies via numerical solutions of semidefinite programs. Our findings highlight the indispensable role of non-stabilizer resources and the inherent limitations of classically simulable operations in quantum state purification, emphasizing the necessity of harnessing the full power of quantum operations for more robust quantum information processing.
\end{abstract}

\maketitle

\textit{Introduction}.---
Quantum state purification serves as a cornerstone technique in quantum information processing, addressing the persistent challenge of noise and decoherence that threaten quantum systems. This process transforms multiple copies of unknown noisy quantum states into purified versions closer to their principal eigenstate, thereby preserving the quantum coherence and fidelity essential for reliable quantum operations. By systematically extracting higher-quality states from multiple noisy copies, purification protocols like entanglement distillation and magic state distillation enable quantum technologies to function effectively despite environmental interference. Universal purification stands out as particularly valuable, offering the ability to improve arbitrary input states without requiring prior knowledge of their properties. This capability extends the practical utility of quantum systems across diverse applications, including quantum computing, communication, and cryptography, making purification an indispensable element of the quantum technology toolkit.

A number of state purification tasks and protocols have been proposed, each striving for various optimality criteria under specific resource and noise models. A prototypical example is the entanglement purification~\cite{bennett1996purification, bennett1996mixed, deutsch1996quantum, dur2007entanglement, horodecki2009quantum}, where multiple copies of noisy entangled states can distill a smaller number of high-fidelity entangled pairs using only local operations and classical
communication (LOCC) or other non-entangling operations. Universal entanglement purification has been studied recently~\cite{zang2025no}, where they prove that only LOCC cannot accomplish universal entanglement purification, while it is possible using more general positive partial-transpose-preserving (PPT) operations. Inspired by some of these works, general state purification under a depolarizing noise was studied~\cite{cirac1999optimal, keyl2001rate, fiuravsek2004optimal, fu2016quantum}, achieving an optimal purification procedure based on different criteria. Recent studies realizing optimal purification protocol utilizing swap-test~\cite{childs2025streaming, grier2025streaming} or controlled-permutation~\cite{yao2024protocols} offer applicable tools for their implementations in quantum devices.

Fault-tolerant quantum computation (FTQC) is another strategy that mitigates decoherence and hardware imperfections by encoding quantum information using quantum error-correcting codes~\cite{shor1995scheme, steane1996error}, some of which have been realized on superconducting qubits and neutral atom platforms~\cite {acharya2024quantum, gupta2024encoding, bluvstein2024logical, rodriguez2024experimental, xu2024constant, bluvstein2025architectural}. 
Stabilizer operations, i.e., Clifford gates, Pauli measurements, and stabilizer state preparations, are central to FTQC. Their foundational properties are captured by the Gottesman-Knill theorem~\cite{gottesman1997stabilizer}, demonstrating that they can be efficiently simulated on a classical computer. To attain the power of universal quantum computation, one must supplement stabilizer operations with magic states or other magic resources~\cite{gottesman1999demonstrating, zhou2000methodology, bravyi2005universal}. The utility of the magic state motivates the development of a quantum resource theory of magic where the free operations are the stabilizer operations, and the free states are the stabilizer states~\cite{wang2019quantifying, seddon2019quantifying, howard2017application, veitch2014resource, chen2025physical, seddon2021quantifying}. Building on these advancements, it is natural to investigate whether the magic resources, with their inherent complexity and cost, are necessary for universal quantum state purification, or whether classically simulable operations alone can achieve this goal. 

In this work, we study this question by developing a framework of quantum state purification under free operations without magic resources. We consider probabilistic state purification~\cite{fiuravsek2004optimal} and define the maximal non-magic purification fidelity that can be computed using semidefinite programming~(SDP). For two-copy inputs, we prove no-go theorems for universal purification via classically simulable operations, i.e., classically simulable operations cannot increase even arbitrarily small fidelity for the purification of the universal state set, no matter what the success probability of the protocol is. We further numerically demonstrate this infeasibility of purification for 3- and 4-copy inputs. 
Furthermore, in the purification for some small sets of states, we numerically show the gaps of the maximal fidelity between classically simulable operations and general quantum operations. The framework of our results is illustrated in Fig.~\ref{fig:framework}.
Our findings reveal the inherent limitations of stabilizer resources in state purification and advance our insights into the border of quantum and classical domains. 

\begin{figure}[t]
    \centering
    \includegraphics[width=\linewidth]{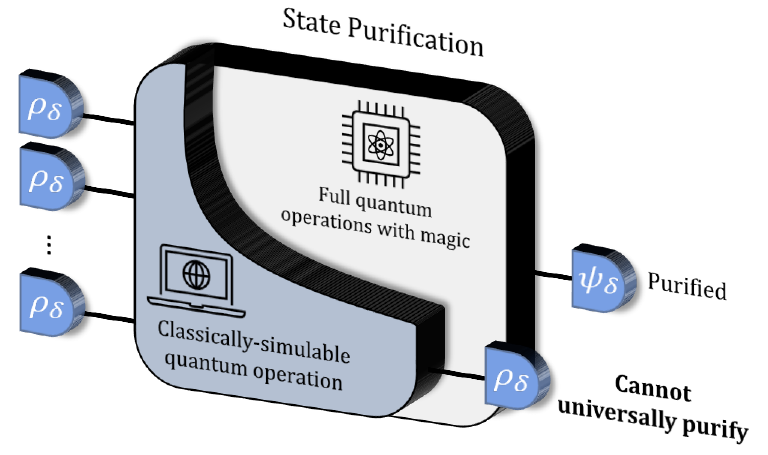}
    \caption{The framework of universal probabilistic purification. The task is to purify multiple copies of a noisy quantum state $\rho_\delta$ into a purer state $\sigma$ using a purification protocol. Protocols with magic are able to conduct universal purification, while classically simulable operations fail to achieve the goal. }
    \label{fig:framework}
\end{figure}

\textit{Preliminaries}. ---
Let $\cH_A$ be a $d$-dimensional Hilbert space with computational basis $\{\ket{j} \}_{j=0,\cdots,d-1}$. A quantum channel $\cN_{A\to B}$ is a linear, completely positive and trace preserving (CPTP) map. We also work with quantum operations that are completely positive and trace non-increasing (CPTN). 
We denote $\cD_\delta(\cdot) = (1-\delta)(\cdot) + \delta I_d/d$ as a depolarizing channel with error parameter $\delta$.

For prime dimension $d$, the Heisenberg-Weyl operators are defined as $T_{\mathbf{u}} \coloneqq T_{a_1, a_2} = \tau^{-a_1a_2}Z^{a_1}X^{a_2},$ where $\tau=e^{(d+1)\pi i/d}$, $X$ and $Z$ are unitary boost and shift operators respectively, and $(a_1, a_2)\in \mathbb{Z}_d\times \mathbb{Z}_d$. The phase space point operator $A^{\mathbf{u}}$ is defined as 
    $A^{\mathbf{0}} \coloneqq\frac{1}{d}\sum_\mathbf{u}T_{\mathbf{u}}, \quad A^{\mathbf{u}} \coloneqq T_{\mathbf{u}}A_{\mathbf{0}}T_{\mathbf{u}}^\dagger.$
The discrete Wigner function of a state $\rho$ is then defined as
$W_{\rho}(\mathbf{u})\coloneqq\frac{1}{d}\tr [A^{\mathbf{u}}\rho].$ For odd dimension, a pure state is a stabilizer state if and only if its discrete Wigner functions are non-negative~\cite{gross2006hudson}.
A quantum channel is called completely positive Wigner preserving (CPWP) if it maps states with non-negative discrete Wigner functions to themselves~\cite{wang2019quantifying}. 

For the qubit case, however, the discrete phase space approach is unsuitable as magic witness~\cite{delfosse2015wigner, raussendorf2017contextuality, raussendorf2020phase}.
Thus, we instead consider another set of stabilizer operations. A quantum channel is called completely stabilizer preserving operations (CSPO) if it maps stabilizer states to stabilizer states~\cite{seddon2019quantifying}. Both CPWP operation and CSPO are standard free operations in the magic resource theory. For odd $d$, CSPO is strictly contained in CPWP~\cite{wang2019quantifying}, but it is hard to witness CSPO except for the qubit case. In what follows, we use CPWP operations (odd $d$) and CSPO ($d=2$) as our free-operation classes. We will show that CSPO and CPWP operations share surprising similarity to some extent on the universal state purification task. More details of both operations can be found in Appendix~\ref{appendix: preliminaries}.  

\textit{No-go theorems for non-magic universal probabilistic purification}.---
The fundamental challenge of noise and decoherence in quantum systems necessitates robust strategies to preserve and enhance quantum information. One such indispensable strategy is quantum state purification, whose general goal is to transform multiple copies of an unknown noisy state into a single, higher-fidelity state approximating the ideal pure form through a specific protocol. In this work, we aim to explore the capability of non-magic operations in the task of state purification under the depolarizing noise $\cD_\delta(\cdot) = (1-\delta)(\cdot) + \delta I_d/d$. We introduce the set of interest $\cA_{\text{CSPO}}\coloneqq \{\text{CSPO}\cap \text{CPTN} \}$, and $\cA_{\text{CPWP}}\coloneqq \{\text{CPWP}\cap \text{CPTN} \}$, and define the set of non-magic purification protocols $\cA$ as
\begin{equation}
    \cA\coloneqq \left\{ 
    \begin{aligned}
        & \cA_{\text{CSPO}},~ \text{(for qubit)}, \\
        &\cA_{\text{CPWP}},~ \text{(for odd dimensional qudit)}. \\
    \end{aligned}
    \right.
\end{equation}

Supposing that we have multiple copies of an unknown pure state undergoing the depolarizing noise $\cD_\delta$, we would like to obtain a state with high fidelity to the ideal pure state for a certain probability. Let $\ket{\psi}$ be the ideal pure state and denote $\psi=\ketbra{\psi}{\psi}$ as its density matrix. We try to find a protocol $\cE_{A^n \to A}$ to purify the noisy states $\cE_{A^n \to A}(\cD_\delta(\psi)^{\ox n})$. Assume that the pure state $\psi_i$ is from a set of pure states $\Psi = \{\psi_1,\psi_2,\cdots  \}$, and denote the size of the set as $|\Psi|$. 
Considering the probabilistic quantum state purification in~\cite{fiuravsek2004optimal}, for a given purification protocol $\cE_{A^n \to A}$, the average purification fidelity of $\Psi$ is 
\begin{equation}
    F(\cE;\Psi)  
    =  \frac{\sum_{\psi_i \in \Psi}\tr[\sigma_{\psi_i} \psi_i]}{\sum_{\psi_i \in \Psi}\tr[\sigma_{\psi_i}]}.
\end{equation}
where $\sigma_{\psi_i}=\cE_{A^n \to A}(\cD_\delta(\psi_i)^{\ox n})$ is the unnormalized post-selected state, and $\cE_{A^n \to A}$ is considered in the set of  non-magic purification protocols $\cA$. For a protocol to be an effective purification, it must achieve an average purification fidelity $F$ that is higher than the initial average fidelity of the noisy input states, $F(\cE;\Psi) > \frac{1}{|\Psi|}\sum_{\psi_i \in \Psi}\tr[\cD_\delta(\psi_i) \psi_i]$. To evaluate the average fidelity for the best non-magic purification protocol, we introduce the \textit{maximal non-magic purification fidelity of $\Psi$} as follows: 
\begin{equation}\label{eq:max_fidelity}
    \begin{aligned}
        F_{\cD_\delta}^{\cA}(n, p; \Psi)\coloneqq  \max_{\cE_{A^n \to A}} & \; \frac{1}{p|\Psi|}{\sum_{\psi_i \in \Psi}\tr[\sigma_{\psi_i} \psi_i]} \\
        {\rm s.t.} & \; \frac{1}{|\Psi|} \sum_{\psi_i \in \Psi}\tr[\sigma_{\psi_i}] = p, \cE_{A^n \to A} \in \cA,
    \end{aligned} 
    \end{equation}
where $p$ is the success probability of the purification. This definition quantifies the highest average fidelity achievable for a given set of noisy input states when purification is restricted to the previously defined non-magic operations in $\cA$. The definition extends to the continuous state set by replacing the summation with the integral over the measure of states. $F_{\cD_\delta}^{\cA}(n, p; \Psi)$ can be formulated as a semidefinite programm~(SDP)~\cite{yao2024protocols}. More formal definitions and derivations can be found in Appendix~\ref{appendix:sdp_max_fid}.

From the definition above, the performance of non-magic purification depends largely on the choice of input state set $\Psi$ as well as the success probability $p$. We first consider the set of all pure states equipped with the Haar measure. We refer to the corresponding task as \textit{non-magic universal
probabilistic purification}, and denote by $F^{\cA}_{\cD_{\delta}}(n,p)$ the value of  Eq.\eqref{eq:max_fidelity} on this set and refer to it as \textit{maximal non-magic universal purification fidelity}. By averaging over the Haar measure, this quantity is independent of any particular ensemble and thus universal.

We focus on the case of two-copy inputs $n=2$. Our first main result demonstrates that there is no such non-magic protocol for effective universal state purification, showing that classically simulable operations make no effort to this task.

\begin{theorem}{\rm (No-go theorem for universal probabilistic purification via CPWP operations)}\label{thm:nogo of CPWP puri} 
There is no two-to-one universal probabilistic purification protocol using CPWP operations for qudit states with odd $d$ under depolarizing noise.
\end{theorem}

We also show that CSPOs fail to conduct universal purification for the single-qubit scenario.

\begin{theorem}{\rm (No-go theorem for universal probabilistic purification via CSPOs)}\label{thm:nogo of CSPO puri}
There is no two-to-one universal probabilistic purification protocol using CSPOs for noisy qubit states under depolarizing noise.
\end{theorem}

We sketch the proof of Theorem~\ref{thm:nogo of CPWP puri} and Theorem~\ref{thm:nogo of CSPO puri} as follows. The core idea is to demonstrate that the maximal non-magic universal purification fidelity $F^{\cA}_{\cD_{\delta}}(2,p)$ always equals the fidelity of the noisy state without any purification $\int \tr[\cD_\delta(\psi)\psi] d\psi = 1 - \frac{d-1}{d}\delta$ for any success probability $p\in (0,1]$. This equality is established by showing that $1 - \frac{d-1}{d}\delta$ serves as both lower bound and upper bound to the SDPs on $F^{\cA}_{\cD_{\delta}}(2,p)$. The full proof for Theorem~\ref{thm:nogo of CPWP puri} is in Appendix~\ref{appendix: proof_CPWP}, and the proof for Theorem~\ref{thm:nogo of CSPO puri} is in Appendix~\ref{appendix: proof_CSPO}.

Theorem~\ref{thm:nogo of CPWP puri} and~\ref{thm:nogo of CSPO puri} indicate that CSPO and CPWP operations offer no improvement for universal state purification with two-copy inputs, no matter what the success probability $p$ is. Since stabilizer operations belong to the classes CPWP and CSPO~\cite{wang2019quantifying, seddon2019quantifying}, we also conclude that stabilizer operations make no effort to this task. We also numerically observe the impossibility of universal purification up to the number of copies $n=4$ for both qubit and qutrit, which is shown in Table~\ref{tab:numerical nogo}. The numerical experiments are implemented in MATLAB~\cite{MATLAB} with interpreters CVX~\cite{cvx} and QETLAB~\cite{qetlab}. The codes are available at~\cite{coderepo}.
\begin{table}[htbp]
\centering
\begin{tabular}{c|ccc}
\hline
& $n=2$ & $n=3$ & $n=4$  \\
\hline 
$d=2(\cA_{\text{CSPO}})$ & \ding{53} & \ding{53} & \ding{53}  \\
$d=3(\cA_{\text{CPWP}})$ & \ding{53} & \ding{53} & \ding{53}  \\
\hline
\end{tabular}
\caption{The numerical results of maximal non-magic universal purification fidelity. "\ding{53}" indicates the impossibilities of universal purification for any error parameters $\delta$ and success probability $p$. }
\label{tab:numerical nogo}
\end{table}
Based on the observation, we conjecture that the universal state purification for any dimension, any number of copies, is impossible without dynamic magic resources. For comparison, the universal purification without any resource constraint is achievable and optimal by projecting onto the symmetric subspace~\cite{childs2025streaming, grier2025streaming, yao2024protocols} through swap-test or more generalized controlled-permutation operations that are known to be magic~\cite{nielsen2010quantum}.

The impossibility of universal purification with classically simulable operations has direct engineering consequences. In practical quantum devices, operations that can be implemented using only stabilizer circuits generally exhibit lower error rates and time costs compared to operations that consume magic. This is primarily because, in current noisy intermediate-scale quantum~(NISQ) devices, arbitrary quantum gates are often decomposed into Clifford gates and elementary non-Clifford gates, such as the T-gate, which increases circuit depth and accumulates errors~\cite{dawson2005solovay, ross2014optimal}. For future fault-tolerant quantum computers, while magic state distillation can achieve low logical error rates, the entire protocol of preparing, distilling, and injecting these magic states is generally resource-intensive, requiring a high overhead in physical qubits and quantum gates, making these procedures costly in terms of system resources and overall error budget~\cite{jones2013multilevel, gidney2021factor, litinski2019game, litinski2019magic, campbell2017roads}. From a practical standpoint for quantum error mitigation, if universal purification could be achieved using only non-magic operations, it would offer advantages by leveraging their generally lower error rates and time costs for noise reduction, and potentially lower error propagation. 
Our results rule out this possibility, emphasizing the role
of magic as a vital resource for quantum computation. The inability of CSPO and CPWP operations to achieve universal purification definitively points to the quantum advantage that magic provides.

\textit{Purifying discrete ensembles with classically simulable operations.}---
The universal no-go results above show that, for Haar random pure inputs subject to depolarizing noise, any post-selected purification protocol built from non-magic operations cannot raise the output fidelity above the input baseline. This settles the universal purification question in the negative. However, universal impossibility does not preclude ensemble-specific purification. In many realistic settings, the relevant states are not arbitrary but come from specific scenarios such as code spaces, calibration sets, or eigenstates of interest. In such cases, structure and prior information can be exploited by classically simulable operations to reject noise selectively while preserving the states of interest with higher fidelity.

We still consider two-copy inputs $n=2$ purification. 
We illustrate two examples involving qubit and qutrit state purification to clearly observe the limitations imposed by the operations in $\cA_{\text{CSPO}}$ and $\cA_{\text{CPWP}}$. For comparison, and to benchmark the impact of these operational restrictions, we also consider the maximal fidelity without any resource restrictions on the purification protocols. We denote this as $F_{\cD_\delta}(2, p; \Psi)$, where the purifying protocol $\cE_{A^n \to A}$ can be any CPTN. The comparison between $F_{\cD_\delta}^{\cA}(2, p; \Psi)$ and $F_{\cD_\delta}(2, p; \Psi)$ will be central to witness the limitations imposed by the non-magic constraint. Figure~\ref{fig:simple_puri_cspo} shows the average fidelities $F_{\cD_\delta}^{\cA}(2, p; \Psi)$ and $F_{\cD_\delta}(2, p; \Psi)$ as functions of the noise parameter $\delta$ for values of the success probability $p=0.1$ and 0.6, in which (a) considers a set of two stabilizer states $\Psi = \{\ketbra{0}{0}, \ketbra{+}{+}\}$ for qubit, and (b) considers the set of four magic states $\Psi = \{\ketbra{\mathbb{S}}{\mathbb{S}}, \ketbra{\mathbb{N}}{\mathbb{N}}, \ketbra{T}{T}, \ketbra{H_+}{H_+}\}$ for qutrit, where
\begin{equation}
\begin{aligned}
&\ket{\mathbb{S}} = \frac{1}{\sqrt{2}}(\ket{1} - \ket{2}), \\ 
&\ket{\mathbb{N}} = \frac{1}{\sqrt{6}}(-\ket{0} + 2\ket{1} - \ket{2}), \\
&\ket{T} = \frac{1}{\sqrt{3}}\Big(e^{2\pi i/9}\ket{0} + \ket{1} + e^{-2\pi i/9}\ket{2}\Big), \\
&\ket{H_+} = \frac{1}{\sqrt{6 + 2\sqrt{3}}}\Big((1+\sqrt{3})\ket{0} + \ket{1} + \ket{2}\Big).
\end{aligned}
\end{equation}
These states are known to exhibit a high degree of magic and are relevant in magic state distillation~\cite{anwar2012qutrit, campbell2012magic, veitch2014resource}. 

Significant gaps are observed between the performances of $\cA_{\text{CSPO}}/\cA_{\text{CPWP}}$ and CPTN operations. CPTN operations consistently achieve higher fidelity than $\cA_{\text{CSPO}}/\cA_{\text{CPWP}}$ across all noise levels and success probabilities. This demonstrates that, even for the two simple state sets, the protocols in $\cA_{\text{CSPO}}/\cA_{\text{CPWP}}$ are fundamentally limited in their ability to purify quantum states. Observe that for both cases, the average fidelities converge when the noise parameter approaches 1, indicating that purification becomes trivial in the presence of maximal noise. However, the convergence values differ due to the limitation of magic resources. Specifically in Figure~\ref{fig:simple_puri_cspo}(a), CPTN purification converges to $\frac{2 + \sqrt{2}}{4} \approx 0.8536$ achieved with the replacement channel $\cE_{A^2_I \to A_O}(\cdot) = \tr(\cdot){\ketbra{\psi}{\psi}}$ where $\ket{\psi} = \cos \frac{\pi}{8}\ket{0} + \sin \frac{\pi}{8}\ket{1}$ is exactly the eigenstate of the Hadamard gate~\cite{bravyi2005universal}. In contrast, $\cA_{\text{CSPO}}$ purification converges to only $\frac{3}{4}$ via the replacement channel $\cE_{A^2_I \to A_O}(\cdot) = \tr(\cdot)\frac{(\ketbra{0}{0} + \ketbra{+}{+})}{2}$. This difference underscores the advantage provided by magic resources, as the CPTN map leverages a magic state to achieve a higher average fidelity.
\begin{figure*}[t]
    \centering
    \includegraphics[width=0.85\linewidth]{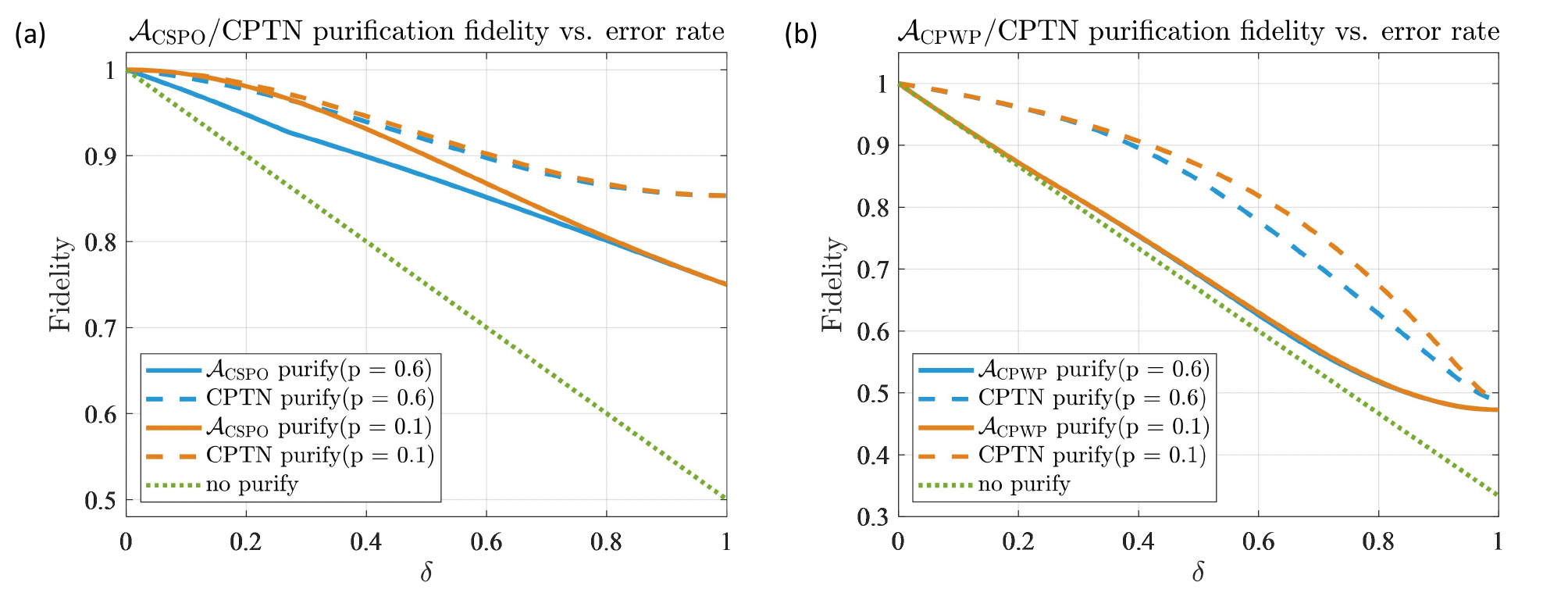}
    \caption{Maximal purification fidelity under depolarizing noise with success probabilities $p=0.1$ and $0.6$. (a) Comparison of average fidelities $F_{\cD_\delta}^{\cA}(2, p; \Psi)$ and $F_{\cD_\delta}(2, p; \Psi)$ for qubit stabilizer states $\Psi = \{\ketbra{0}{0}, \ketbra{+}{+}\}$. (b) Similar comparison for qutrit magic states $\Psi = \{\ketbra{\mathbb{S}}{\mathbb{S}}, \ketbra{\mathbb{N}}{\mathbb{N}}, \ketbra{T}{T}, \ketbra{H_+}{H_+}\}$. In both cases, classically simulable operations (CSPO/CPWP) achieve limited purification but remain strictly less effective than general quantum operations across all noise rates $\delta$.
    }
    \label{fig:simple_puri_cspo}
\end{figure*}

\textit{Concluding remarks}.---
In this work, we have investigated the quantum state purification task within the limitations of classically simulable operations. We formulated maximal non-magic purification fidelity, demonstrating that it can be computed using SDP. We then prove no-go results implying that classically simulable operations make no effort to universal purification for two-copy inputs, regardless of whether the protocol is deterministic or probabilistic, for a qubit system and any odd dimension. Furthermore, our 3- and 4-copy numerical experiments confirm that universal purification requires magic resources and remains infeasible under these restricted operations. We also demonstrate maximal purification fidelity for two small sets of states, showing that general quantum operations achieve better purification performance than classically simulable operations. This distinction underscores the indispensable role of magic resources in realizing state purification.

In contrast, protocols such as the swap-test and controlled-permutation use magic and can achieve universal purification. In fault-tolerant architectures, these non-Clifford steps are realized via magic state injection, which introduces additional noise and overhead determined by the fidelity of the distilled magic states and the injection circuits. A practical implementation is to perform stabilizer operations under error detection and correction, while providing the minimal non-Clifford subroutines through magic state distillation. High-fidelity magic state distillation combined with stabilizer operations supplies the required computational capability for purification while limiting additional noise.

Besides, our results open up several new research directions. The established framework and analytical methodology can be applied to investigate purification under dynamic resource limitations for other specific sets of states, such as pure magic states or entangled states. These tasks share conceptual similarities with established resource purification paradigms, including entanglement distillation and magic state distillation. Applying our approach to these aspects of resource purification may offer new insights into their practical limits.

Another direction involves analyzing universal state purification under different noise models. Our findings demonstrate that classically simulable operations are insufficient to purify states subjected to depolarizing noise. Given that depolarizing noise is generally considered the most challenging for purification methods like the swap-test~\cite{grier2025streaming}, this result highlights a significant limitation. However, this does not preclude their potential effectiveness against other noise types. For instance, recent work indicates that the amplitude damping (AD) channel can generate magic~\cite{trigueros2025nonstabilizerness}. Investigating non-magic purification protocols for the AD channel, and contrasting them with channels that do not generate magic, could elucidate the interplay between noise properties, magic generation, and purification feasibility.

Finally, this work contributes to the broader understanding of the resource theory of magic. While the resource theory of entanglement has significantly impacted quantum information processing and quantum communication, connections between the resource theory of magic and other operational tasks are less explored. Previous works have examined the impact of non-magic operations on quantum channel capacity~\cite{bu2025magic} and quantum state discrimination~\cite{zhu2024limitations}. Our findings provide a concrete example where magic is necessarily required for universal state purification. Exploring the capabilities and limitations of magic operations in a wider range of quantum information processing tasks remains an important area of study.

\textit{Acknowledgements}.---
We would like to thank Yu-Ao Chen and Bartosz Regula for the helpful discussion. This work was supported by the National Key R\&D Program of China (Grant No.~2024YFE0102500), the National Natural Science Foundation of China (Grant. No.~12447107), the Guangdong Provincial Quantum Science Strategic Initiative (Grant No.~GDZX2403008, GDZX2403001), the Guangdong Provincial Key Lab of Integrated Communication, Sensing and Computation for Ubiquitous Internet of Things (Grant No. 2023B1212010007), the Quantum Science Center of Guangdong-Hong Kong-Macao Greater Bay Area, and the Education Bureau of Guangzhou Municipality.

\bibliography{main}

\onecolumngrid
\newpage

\appendix
\setcounter{subsection}{0}
\setcounter{table}{0}
\setcounter{figure}{0}

\vspace{3cm}

\begin{center}
\Large{\textbf{Supplemental Material} \\ \textbf{
}}
\end{center}

\renewcommand{\theequation}{S\arabic{equation}}
\renewcommand{\theproposition}{S\arabic{proposition}}
\renewcommand{\thedefinition}{S\arabic{definition}}
\renewcommand{\thefigure}{S\arabic{figure}}
\setcounter{equation}{0}
\setcounter{table}{0}
\setcounter{section}{0}
\setcounter{proposition}{0}
\setcounter{definition}{0}
\setcounter{figure}{0}

In this Supplemental Material, we provide detailed proofs of the theorems and propositions in the manuscript ``No-Go Theorems for Universal Quantum State Purification via Classically Simulable Operations''. In Appendix~\ref{appendix: preliminaries}, we cover the notation in this work and preliminaries of the discrete Wigner function and the robustness of magic as the magic witness for odd qudit systems and qubit systems, respectively, and introduce the CPWP operation and CSPO. In Appendix~\ref{appendix:sdp_max_fid}, we give a formal definition of maximal non-magic purification fidelity and its corresponding SDP problems. In Appendix~\ref{appendix: proof_CPWP} and Appendix~\ref{appendix: proof_CSPO}, we prove Theorem~\ref{thm:nogo of CPWP puri} and~\ref{thm:nogo of CSPO puri} respectively, which are the no-go theorems or universal probabilistic purification via CPWP operations and CSPOs. 


\section{Preliminaries}\label{appendix: preliminaries}

We consider a finite-dimensional Hilbert space $\cH_A$ representing the quantum system $A$ with dimension $d$. Let $\{\ket{j} \}_{j=0,\cdots,d-1}$ be a standard computational basis. We use $\cL(\cH_A)$ to represent the set of linear operators that map from $\cH_A$ to itself. A density operator is a positive semidefinite operator in $\cL(\cH_A)$ with trace one, and $\cD(\cH_A)$ denotes the set of all density operators in $\cH_A$. A quantum channel $\cN_{A\to B}$ is a linear map from $\cL(\cH_A)$ to $\cL(\cH_{B})$ that is completely positive and trace-preserving (CPTP). We also introduce a quantum operation that is completely positive and trace non-increasing (CPTN). Their associated Choi-Jamio\l{}kowski operators are $J^{\cN}_{AB}\coloneqq \sum_{i, j=0}^{d-1}\ketbra{i}{j} \ox \cN_{A \to B}(\ketbra{i}{j})$~\cite{choi1975completely, jamiolkowski1972linear}. We denote $\cI$ as the identity channel, and $\cD_\delta$ as a depolarizing channel with error parameter $\delta$. The symmetric group of degree $n$ is denoted by $\cS_n$  and $\mathbf{P}_n(c)$ represents the permutation operator for $c\in\cS_n$. We denote $\Pi_n\coloneqq\frac{1}{n!}\sum_{c\in\cS_n}\mathbf{P}_n(c)$ as the projector on the symmetric subspace of $\cH_{A}^{\ox n}$.

\subsection{The discrete Wigner function}
Let us introduce the basics of the discrete Wigner function for quantum information theory~\cite{gross2006hudson,veitch2014resource, wang2019quantifying}.
For prime dimension $d$, the unitary boost and shift operators $X, Z\in \cL(\cH)$ are defined as: 
\begin{equation}
    X\ket{j}=\ket{j\oplus1}, \quad Z\ket{j}=w^j\ket{j},
\end{equation}
where $w=e^{2\pi i/d}$ and $\oplus$ denotes addition modulo $d$. The Heisenberg-Weyl operators are defined as 
\begin{equation}
    T_{\mathbf{u}} = \tau^{-a_1a_2}Z^{a_1}X^{a_2},
\end{equation}
where $\tau=e^{(d+1)\pi i/d}$, $\mathbf{u}=(a_1, a_2)\in \ZZ_d \times \ZZ_d$. For the composite system $\cH_A \ox \cH_B$, the Heisenberg-–Weyl operators are the tensor product of the subsystem Heisenberg–-Weyl operators:
\begin{equation}
    T_{\mathbf{u}_A \oplus \mathbf{u}_B} = T_{\mathbf{u}_A} \ox T_{\mathbf{u}_B},
\end{equation}
where $\mathbf{u}_A \oplus \mathbf{u}_B \in \ZZ_d \times \ZZ_d \times \ZZ_d \times \ZZ_d$. For each point $\mathbf{u}$ in the discrete phase space, there is a corresponding phase space point operator $A^{\mathbf{u}}$ defined as 
\begin{equation}
    A^{\mathbf{0}} \coloneqq\frac{1}{d}\sum_\mathbf{u}T_{\mathbf{u}}, \quad A^{\mathbf{u}} \coloneqq T_{\mathbf{u}}A_{\mathbf{0}}T_{\mathbf{u}}^\dagger.
\end{equation}
The discrete Wigner function of a state $\rho$ at the point $\mathbf{u}$ is then defined as
\begin{equation}
    W_{\rho}(\mathbf{u})=\frac{1}{d}\tr [A^{\mathbf{u}}\rho].
\end{equation}
More generally, we can replace $\rho$ with a Hermitian operator $H$ for the discrete Wigner function. Some useful properties are listed:
\begin{enumerate}
    \item $A^{\mathbf{u}}$ is Hermitian;
    \item $\sum_\mathbf{u}A^\mathbf{u}/d=I$;
    \item $\tr[A^\mathbf{u}A^{\mathbf{u'}}]=d\delta(\mathbf{u}, \mathbf{u'})$, where $\delta(a,b)$ is the discrete Dirac delta function;
    \item $\tr[A^\mathbf{u}]=1$;
    \item $H = \sum_{\mathbf{u}}W_H(\mathbf{u})A^\mathbf{u}$;
    \item $\{A^\mathbf{u}\}_\mathbf{u} = \{(A^\mathbf{u})^T\}_\mathbf{u}$;
    \item For the composite system $\cH_A \ox \cH_B$, the phase space point operators are the tensor product of the subsystem phase space point operators $A^{\mathbf{u}_A\oplus \mathbf{u}_B} = A^{\mathbf{u}_A} \ox A^{\mathbf{u}_B}$ ;
\end{enumerate}

A Hermitian operator $H$ has non-negative discrete Wigner functions if $\forall \mathbf{u}, W_H(\mathbf{u})\geq 0$. For odd dimensions, according to the discrete Hudson's theorem~\cite{gross2006hudson}, a pure state is a stabilizer state if and only if it has non-negative discrete Wigner functions. We denote the set of quantum states with non-negative discrete Wigner functions by $\cW_+$ (Wigner polytope)
\begin{equation}
    \cW_+ \coloneqq \big\{\rho:\forall \mathbf{u}, W_\rho(\mathbf{u})\geq 0, \rho \geq 0, \tr\rho=1  \big\}.
\end{equation}

The discrete Wigner function of a state can be naturally extended to quantum channels. We recall the definition of the discrete Wigner function of a quantum channel from~\cite{mari2012positive} (see the discussion surrounding~\cite[Eq.~(10)]{mari2012positive}), which is related to the Wigner function of a quantum channel as defined in~\cite[Eq.~(95)]{Bartlett_2012}, and formally defined in~\cite{wang2019quantifying}.

\begin{definition}{\rm (Discrete Wigner function of a quantum channel~\cite{wang2019quantifying})}
Given a quantum channel $\cN_{A\rightarrow B}$, its discrete Wigner function is defined as 
\begin{equation}
\begin{aligned}
    \cW_{\cN}(\mathbf{v}|\mathbf{u}) &\coloneqq \frac{1}{d_B}\tr \big[((A^\mathbf{u}_A)^T \ox A^\mathbf{v}_B) J^\cN_{AB}\big] \\
    &= \frac{1}{d_B}\tr \big[A^\mathbf{v}_B\cN(A^\mathbf{u}_A)\big].
\end{aligned}
\end{equation}
\end{definition}
A quantum circuit consisting of operations with non-negative Wigner functions can be classically simulated~\cite{pashayan2015estimating}. The completely positive Wigner preserving operations are considered as free operations for the resource theory of magic states. 
\begin{definition}{\rm (CPWP operation~\cite{wang2019quantifying})}\label{def:cpwp operations} 
A Hermiticity preserving linear map $\cN_{A\rightarrow B}$ is called completely positive Wigner preserving (CPWP) if for any system $R$ with odd dimension, the following holds
\begin{equation}
    ({\rm \cI}_R \ox \cN_{A\to B})(\rho_{RA})\in \cW_+, \quad \forall \rho_{RA} \in \cW_+.
\end{equation}
\end{definition}
It was shown that a quantum channel is CPWP if and only if its discrete Wigner functions are non-negative~\cite[Theorem~2]{wang2019quantifying}.

\subsection{The robustness of magic}
For the qubit case, however, the discrete phase space approach cannot be used to quantify the magic of quantum operation since it consists of some Clifford operations from free operations~\cite{delfosse2015wigner, raussendorf2017contextuality}, or is not able to build up representations under the tensor product~\cite{raussendorf2020phase}. Thus, we refer to the magic witness in terms of stabilizer states. We denote  $\mathrm{STAB}_n$ as the set of all $n$-qubit stabilizer states. 
Refs.~\cite{howard2017application, seddon2019quantifying} introduce a scheme to decompose density matrices as real linear combinations of pure stabilizer state projectors.

The pure states in $\mathrm{STAB}_n$ construct an overcomplete basis for the set of $2^n$-dimensional density matrices, where $n$ is the number of qubits. Any density matrix can be decomposed as an affine combination of pure stabilizer state projectors, i.e.,
\begin{equation}
    \rho = \sum_j x_j \ketbra{\phi_j}{\phi_j}, \quad \sum_jx_j=1, \quad \ketbra{\phi_j}{\phi_j}\in \mathrm{STAB}_n.
\end{equation}
Robustness of magic is defined as the minimal $l_1$-norm $\|\mathbf{x}\|_1 = \sum_j |x_j|$ over all possible decompositions~\cite{howard2017application}, i.e.,
\begin{equation}
    \cR(\rho) = \min_\mathbf{x}\Big\{ \|\mathbf{x}\|_1: \sum_j x_j\ketbra{\phi_j}{\phi_j}=\rho,~\ketbra{\phi_j}{\phi_j}\in \mathrm{STAB}_n\Big\}.
\end{equation}
Equivalently, it is given by
\begin{equation}
    \cR(\rho) = \min_{\rho_{\pm} \in \mathrm{STAB}_n} \big\{1+2p: (1+p)\rho_+ - p\rho_-=\rho, ~p \geq 0 \big\}.
\end{equation}
The optimization of the robustness can be written in terms of a linear system as~\cite{howard2017application}
\begin{equation}\label{eq: LP_robustness of magic}
    \cR(\rho) = \min_{\mathbf{x}} \big\{\|\mathbf{x}\|_1:   G\mathbf{x}=b \big\},
\end{equation}
where $G_{ij} = \tr [P_i \ketbra{\phi_j}{\phi_j}], \ketbra{\phi_j}{\phi_j} \in \mathrm{STAB}_n$, $b_i=\tr[P_i \rho]$,  where $P_i$ is the $i$-th Pauli operators given the number of qubits. \
The robustness of magic is a magic monotone, with the property of faithfulness: if $\rho\in$ $\mathrm{STAB}_n$, then $\cR(\rho)=1$ with every $0\leq x_j \leq1, \forall x_j$; otherwise $\cR(\rho)>1$. Note that this holds for an unnormalized state:
if a state has $\tr\rho = p$, then $\rho\in \mathrm{STAB}_n$ if and only if $\cR(\rho)=p$ with every $0\leq x_j \leq p, \forall x_j$. We use the tools~\cite{Seddon2022channelmagic} to calculate the linear programming Eq.~\eqref{eq: LP_robustness of magic} (see also~\cite{seddon2019quantifying, seddon2022advancing}).

The above definition of robustness of magic is also naturally extended to quantify quantum operations, where the free operations are characterized to map stabilizer states to stabilizer states~\cite{seddon2019quantifying}.

\begin{definition}{\rm (Completely stabilizer preserving operations~\cite{seddon2019quantifying})}\label{def:cspo operations} 
A CPTP map $\cN_{A\rightarrow B}$ is called completely stabilizer-preserving if for any system $R$, the following holds
\begin{equation}
    ({\rm \cI}_R \ox \cN_{A\to B})(\rho_{RA})\in {\rm STAB}_{m + n}, \quad \forall \rho_{RA} \in {\rm STAB}_{m + n},
\end{equation}
where $m,n$ are number of qubits for systems $R$ and $A$ respectively. 
\end{definition}
We abbreviate the completely stabilizer preserving operation as CSPO or CSPOs for convenience. It was proved that a quantum channel is a CSPO if and only if its Choi state $\Phi_{AB}^{\cN} = J_{AB}^{\cN}/d$ is a stabilizer state~\cite[Theorem~3.1]{seddon2019quantifying}. The channel robustness of magic $\cR_*(\cN)$ is then defined as 
\begin{equation}
    \cR_*(\cN) = \min_{\Lambda_{\pm}\in \text{CSPO}} \big\{2p+1: (1+p)\Lambda_+ - p\Lambda_-=\cN, ~p\geq 0 \big\}.
\end{equation}
An equivalent definition in terms of the Choi state is
\begin{equation}
    \cR_*(\cN) = \min_{\rho_{\pm} \in \mathrm{STAB}_{2n}} \Big\{2p+1: (1+p)\rho_+ - p\rho_-=\Phi_{AB}^{\cN}, ~p \geq 0,
    ~ \tr_B[\rho_{\pm}]=\frac{I_A}{d_A} \Big\}.
\end{equation}
This can also be computed using linear programming Eq.~\eqref{eq: LP_robustness of magic}. Channel robustness also satisfies the faithfulness: if $\cN$ is a CSPO, then $\cR_*(\cN)=1$; otherwise $\cR_*(\cN)>1$. Note that if an operation $\cN$ is trace non-increasing $\tr_B[\Phi_{AB}^{\cN}] \leq pI_A/d_A $, then $\cN\in$ CSPO if and only if $\cR_*(\cN) = p$ with every $x_j$ in the linear programming Eq.~\eqref{eq: LP_robustness of magic} satisfying $0\leq x_j \leq p, \forall x_j$. 

\section{Formal definitions of maximal non-magic purification fidelity and its SDPs}\label{appendix:sdp_max_fid}

We introduce the set of interest $\cA_{\text{CSPO}}\coloneqq \{\text{CSPO}\cap \text{CPTN} \}$, and $\cA_{\text{CPWP}}\coloneqq \{\text{CPWP}\cap \text{CPTN} \}$, and define the set of non-magic purification protocols $\cA$ as
\begin{equation}
    \cA\coloneqq \left\{ 
    \begin{aligned}
        & \cA_{\text{CSPO}},~ \text{(for qubit)}, \\
        &\cA_{\text{CPWP}},~ \text{(for odd dimensional qudit)}. \\
    \end{aligned}
    \right.
\end{equation}
Let $\ket{\psi}$ be the ideal pure state and denote $\psi=\ketbra{\psi}{\psi}$ as its density matrix. We try to find a protocol $\cE_{A^n \to A}$ to purify the noisy states $\cE_{A^n \to A}(\cN(\psi)^{\ox n})$. Assume that the pure state $\psi_i$ is from a discrete set of pure states $\Psi = \{\psi_1,\psi_2,\cdots  \}$, and denote the size of the set as $|\Psi|$. The maximal non-magic purification fidelity of $\Psi$ is defined as
\begin{definition}{\rm (Maximal non-magic purification fidelity)}\label{def:max_fidelity} 
For a finite set of pure states $\Psi = \{\psi_1,\psi_2,\cdots\} \subset \cD(\cH_A)$ of size $|\Psi|$, subject to a $d$-dimensional depolarizing channel $\cD_{\delta}$, the maximal non-magic purification fidelity with success probability $p\in(0,1]$ and $n$ copies of the noisy state is defined as
    \begin{equation}
    \begin{aligned}
        F_{\cD_\delta}^{\cA}(n, p; \Psi)\coloneqq  \max_{\cE_{A^n \to A}} & \; \frac{1}{p|\Psi|}{\sum_{\psi_i \in \Psi}\tr[\sigma_{\psi_i} \psi_i]} \\
        {\rm s.t.} & \; \frac{1}{|\Psi|} \sum_{\psi_i \in \Psi}\tr[\sigma_{\psi_i}] = p, \cE_{A^n \to A} \in \cA,
    \end{aligned} 
    \end{equation}
    where $\sigma_{\psi_i}=\cE_{A^n \to A}(\cD_\delta(\psi_i)^{\ox n})$ is the unnormalized post-selected state.   
\end{definition}
$F^{\cA}_{\cD_\delta}(n,p; \Psi)$ can be calculated through semidefinite programming~(SDP). By defining two key operators that will appear in our SDPs:

\begin{equation}\label{eq:discrete Q and R}
\begin{aligned}
     Q_{A^n_IA_O} &\coloneqq \frac{1}{|\Psi|}\sum_{\psi_i\in \Psi}\cD_{\delta}(\psi_i)^{\ox n} \ox \psi_i = (\cD_{\delta}^{\ox n} \ox \cI)\Big(\frac{1}{|\Psi|}\sum_{\psi_i\in \Psi}\psi_i^{\ox {n+1}}\Big),\\
     R_{A^n_IA_O} &\coloneqq \cD_{\delta} \Big(\frac{1}{|\Psi|}\sum_{\psi_i\in \Psi}\psi_i^{\ox n} \Big)\ox I,
\end{aligned}
\end{equation}
the SDP formulations of $F_{\cD_\delta}^{\cA}(n, p; \Psi)$ in odd qudit systems are given as follows.
\begin{proposition}\label{sdp:discrete cpwp}
For a finite set of pure states $\Psi = \{\psi_1,\psi_2,\cdots\} \subset \cD(\cH_A)$ of size $|\Psi|$, subject to an odd $d$-dimensional depolarizing channel $\cD_{\delta}$, the maximal non-magic purification fidelity with success probability $p\in(0,1]$ and $n$ copies of the noisy state is characterized by the following SDPs:
\begin{equation}
\begin{aligned}
&\underline{\textbf{Primal Program}}\\
\max_{J^{\cE}_{A^n_IA_O}} &\;\; \frac{1}{p}\tr\Big[J^{\cE}_{A^n_IA_O}Q_{A^n_IA_O}^{T_{A^n_I}}\Big]\\
 {\rm s.t.}
&\; \tr\Big[J^{\cE}_{A^n_IA_O}R_{A^n_IA_O}^{T_{A^n_I}}\Big]=p, \\
&\; \tr_{A_O}[J^{\cE}_{A^n_IA_O}]\leq I_{A^n_I}, \;J^{\cE}_{A^n_IA_O}\geq 0, \\
&\;  \tr\Big[\left((A^{\mathbf{u}}_{A^n_I})^T \ox A^{\mathbf{v}}_{A_O}\right)J^{\cE}_{A^n_IA_O}\Big] \geq 0, \; \forall \mathbf{u},\mathbf{v},
\end{aligned}
\begin{aligned}
&\underline{\textbf{Dual Program}}\\
\min_{x, Y_{A^n_I}, c_{\mathbf{u},\mathbf{v}}} &\,  -x - \frac{1}{p}\tr[Y_{A^n_I}]\\
{\rm s.t.} 
&\;Q_{A^n_IA_O}^{T_{A^n_I}} + xR_{A^n_IA_O}^{T_{A^n_I}} + Y_{A^n_I}\ox I_{A_O} \leq C^{T_{A^n_I}}_{A^n_IA_O}, \\
&\; C_{A^n_IA_O} = \sum_{\mathbf{u},\mathbf{v}}c_{\mathbf{u},\mathbf{v}} A^{\mathbf{u}}_{A^n_I} \ox A^{\mathbf{v}}_{A_O}, \\
&\; Y_{A^n_I} \leq 0, \; c_{\mathbf{u},\mathbf{v}}\leq 0, \; \forall \mathbf{u},\mathbf{v},
\end{aligned}
\end{equation}
where $J^{\cE}_{A^n_IA_O}$ denotes the Choi operator of $\cE_{A^n \to A}$ in $\cA_{\text{CPWP}}$, $T_{A^n_I}$ denotes the partial transpose operation on system $A^n_I$, and $Q_{A^n_IA_O}, R_{A^n_IA_O}$ are given in Eq.~\eqref{eq:discrete Q and R}.  $A^{\mathbf{u}}_{A_I^n}$ and 
$A^{\mathbf{v}}_{A_O}$ denote the $n$-copy and 1-copy phase space point operators, respectively.
\end{proposition}

\begin{proof}
We first derive the primal SDP from the Definition~\ref{def:max_fidelity} for an odd qudit system, then derive its dual program.
The objective function of the primal SDP is calculated as
\begin{equation}\label{eq:objective discrete sdp}
\begin{aligned}
    \frac{1}{p|\Psi|}\sum_{\psi_i \in \Psi}\tr[\sigma_{\psi_i} \psi_i] &= \frac{1}{p|\Psi|}\sum_{\psi_i \in \Psi} \tr\big[\cE_{A^n \to A}(\cD_\delta(\psi_i)^{\ox n}) \psi_i \big] \\
    &= \frac{1}{p|\Psi|}\sum_{\psi_i \in \Psi} \tr\Big[(J^{\cE}_{A^n_IA_O})^{T_{A_I^n}} \big(\cD_\delta(\psi_i)^{\ox n} \ox \psi_i\big)\Big] \\
    &= \frac{1}{p}\tr \left[(J^{\cE}_{A^n_IA_O})^{T_{A_I^n}}\frac{1}{|\Psi|}\sum_{\psi_i \in \Psi} \cD_\delta(\psi_i)^{\ox n} \ox \psi_i  \right] \\
    &= \frac{1}{p}\tr\Big[J^{\cE}_{A^n_IA_O}Q_{A^n_IA_O}^{T_{A_I^n}}\Big],
\end{aligned}
\end{equation}
and similarly $\frac{1}{|\Psi|} \sum_{\psi_i \in \Psi}\tr[\sigma_{\psi_i}] = \tr\Big[(J^{\cE}_{A^n_IA_O})^{T_{A_I^n}} R_{A^n_IA_O}\Big] = p$. Furthermore, the constraints are derived from the fact that $\cE_{A^n \to A}$ is in $\cA_{\text{CPWP}}$ where $\tr_{A_O}[J^{\cE}_{A^n_IA_O}]\leq I_{A^n_I},\;J^{\cE}_{A^n_IA_O}\geq 0$ imply the Choi matrix of a CPTN map, and $\tr\Big[\left((A^{\mathbf{u}}_{A^n_I})^T \ox A^{\mathbf{v}}_{A_O}\right)J^{\cE}_{A^n_IA_O}\Big] \geq 0$ imply the non-negative discrete Wigner functions of a quantum operation. We complete this calculation for the primal program.

To derive the dual SDP, the Lagrange function can be written as
\begin{equation}
\begin{aligned}
    &\cL(x, Y_{A^n_I}, c_{\mathbf{u},\mathbf{v}}, J^{\cE}_{{A^n_I A_O}}) \\
    &\coloneqq \frac{1}{p}\tr\Big[J^{\cE}_{{A^n_I A_O}}Q_{A^n_IA_O}^{T_{A^n_I}}\Big] +\Big\langle x,\tr\Big[J^{\cE}_{{A^n_I A_O}}R_{A^n_I A_O}^{T_{A^n_I}}\Big] - p\Big\rangle + \Big\langle Y_{A^n_I}, \tr_{A_O}[J^{\cE}_{{A^n_I A_O}}]-I_{A^n_I}\Big\rangle\\
    &\quad - \sum_{\mathbf{u}, \mathbf{v}}c_{\mathbf{u},\mathbf{v}} \tr\left[\left((A^\mathbf{u}_{A^n_I})^T \ox A^\mathbf{v}_{A_O}\right)J^{\cE}_{A^n_IA_O}\right] \\
    &=-px-\tr[Y_{A^n_I}]+ \Big\langle J^{\cE}_{{A^n_I A_O}}, \frac{1}{p}Q_{A^n_IA_O}^{T_{A^n_I}}+R_{A^n_IA_O}^{T_{A^n_I}}x+Y_{A_I^n}\ox I_{A_O} - \sum_{\mathbf{u}, \mathbf{v}}c_{\mathbf{u},\mathbf{v}} (A^\mathbf{u}_{A^n_I})^T \ox A^\mathbf{v}_{A_O} \Big\rangle,
\end{aligned}
\end{equation}
where $x, Y_{A^n_I}, c_{\mathbf{u},\mathbf{v}}$ are Lagrange multipliers. Then, the Lagrange dual function can be written as
\begin{align}
     \mathcal{G}(x, &Y_{A^n_I}, c_{\mathbf{u},\mathbf{v}})\coloneqq\sup_{J^{\cE}_{A^n_IA_O}\geq0}\mathcal{L}(x, Y_{A^n_I}, c_{\mathbf{u},\mathbf{v}}, J^{\cE}_{{A^n_I A_O}}).
\end{align}
Since $J^{\cE}_{{A^n_I A_O}}\geq0$, it holds that
\begin{align}
    &\frac{1}{p}Q_{A^n_IA_O}^{T_{A^n_I}}+xR_{A^n_IA_O}^{T_{A^n_I}}+Y_{A_I^n}\ox I_{A_O}  - \sum_{\mathbf{u}, \mathbf{v}}c_{\mathbf{u},\mathbf{v}} (A^\mathbf{u}_{A^n_I})^T \ox A^\mathbf{v}_{A_O}\leq 0.
\end{align}
Otherwise, the inner norm is unbounded. Similarly, we have $Y_{A^n_I}\leq 0$ due to $\tr_{ A}[J^{\cE}_{A^n_I A_O}] \leq I_{A^n_I}$ and $c_{\mathbf{u},\mathbf{v}} \leq 0$ due to $ \tr\left[\left((A^\mathbf{u}_{A^n_I})^T \ox A^\mathbf{v}_{A_O}\right)J^{\cE}_{A^n_IA_O}\right] \geq 0$. Redefine $x$ as $x/p$, $Y_{A^n_I}$ as $Y_{A^n_I}/p$ and $c_{\mathbf{u},\mathbf{v}}$ as $c_{\mathbf{u},\mathbf{v}}/p$, then we obtain the following dual SDP
\begin{equation}
\begin{aligned}
\min_{x, Y_{A^n_I}, c_{\mathbf{u},\mathbf{v}}} &\,  -x - \frac{1}{p}\tr[Y_{A^n_I}]\\
     {\rm s.t.} 
    &\;Q_{A^n_IA_O}^{T_{A^n_I}} + xR_{A^n_IA_O}^{T_{A^n_I}} + Y_{A^n_I}\ox I_{A_O} \leq C^{T_{A^n_I}}_{A^n_IA_O} , \\
    &\; C_{A^n_IA_O} = \sum_{\mathbf{u}, \mathbf{v}}c_{\mathbf{u},\mathbf{v}} A^\mathbf{u}_{A^n_I} \ox A^\mathbf{v}_{A_O}  , \\
    &\; Y_{A^n_I} \leq 0, \; c_{\mathbf{u},\mathbf{v}}\leq 0, \; \forall \mathbf{u}, \mathbf{v},
    \end{aligned}
\end{equation}
It is worth noting that strong duality is held by Slater's condition with $J^{\cE}_{A^n_IA_O} = pI_{A^n_IA_O}/d_{A_O}$. We complete the proof. 
\end{proof} 

The SDP formulations of $F_{\cD_\delta}^{\cA}(n, p; \Psi)$ in a multi-qubit system are similar,
\begin{proposition}\label{sdp:discrete cspo}
For a finite set of pure states $\Psi = \{\psi_1,\psi_2,\cdots\} \subset \cD(\cH_A)$ of size $|\Psi|$, subject to a multi-qubit depolarizing channel $\cD_{\delta}$, the maximal non-magic purification fidelity with success probability $p\in(0,1]$ and $n$ copies of the noisy state is characterized by the following SDPs:
\begin{equation}
\begin{aligned}
&\underline{\textbf{Primal Program}}\\
\max_{J^{\cE}_{A^n_IA_O}, x_j} &\;\; \frac{1}{p}\tr\Big[J^{\cE}_{A^n_IA_O}Q_{A^n_IA_O}^{T_{A^n_I}}\Big]\\
{\rm s.t.}
&\; \tr\Big[J^{\cE}_{A^n_IA_O}R_{A^n_IA_O}^{T_{A^n_I}}\Big]=p, \\
&\; \tr_{A_O}[J^{\cE}_{A^n_IA_O}]\leq I_{A^n_I}, \;J^{\cE}_{A^n_IA_O}\geq 0, \\
&\;  \sum_j G_{ij}x_j = b_i \;, x_j \geq 0, \; \forall j,\\
&\; b_i =  \frac{1}{d_{A^n_I}}\tr\big[P_iJ^{\cE}_{A^n_IA_O}\big],
\\
\end{aligned}
\begin{aligned}
&\underline{\textbf{Dual Program}}\\
\min_{x, Y_{A^n_I}, y_i} &\,  -x - \frac{1}{p}\tr[Y_{A^n_I}]\\
 {\rm s.t.} 
&\;Q_{A^n_IA_O}^{T_{A^n_I}} + xR_{A^n_IA_O}^{T_{A^n_I}} + Y_{A^n_I}\ox I_{A_O} \leq D_{A^n_IA_O}^{T_{A^n_I}}, \\
&\;  D_{A^n_IA_O} =  \sum_iy_i(P^i_{A^n_IA_O})^{T_{A^n_I}}, \\
&\; \sum_iy_iG_{ij} \leq 0 , \; Y_{A^n_I} \leq 0,
\end{aligned}
\end{equation}
where $J^{\cE}_{A^n_IA_O}$ denotes the Choi operator of $\cE_{A^n \to A}$ in $\cA_{\text{CSPO}}$, $T_{A^n_I}$ denotes the partial transpose operation on system $A^n_I$, and $Q_{A^n_IA_O}, R_{A^n_IA_O}$ are given in Eq.~\eqref{eq:discrete Q and R}. $G_{ij} = \tr [P_i \ketbra{\phi_j}{\phi_j}], \ketbra{\phi_j}{\phi_j} \in \mathrm{STAB}_{n+1}$, and $\{P_i\}$ are Pauli operators.
\end{proposition}

\begin{proof} We still derive the primal SDP from the Definition~\ref{def:max_fidelity} for multi-qubit system, then derive its dual program. The calculation of objective function and the constraint of $\tr\big[(J^{\cE}_{A^n_IA_O})^{T_{A_I^n}} R_{A^n_IA_O}\big] = p$ is identical with Eq.~\eqref{eq:objective discrete sdp}. Other constraints are derived from $\cE_{A^n \to A}$ in $\cA_{\text{CSPO}}$ where $\tr_{A_O}[J^{\cE}_{A^n_IA_O}]\leq I_{A^n_I}, \;J^{\cE}_{A^n_IA_O}\geq 0$ imply the Choi matrix of a CPTN map, and $b_i =  \frac{1}{d_{A^n_I}}\tr\big[P_iJ^{\cE}_{A^n_IA_O}\big]$ and $\sum_j G_{ij}x_j = b_i \;, x_j \geq 0, \; \forall j$ imply the constraints of channel robustness of magic for CSPOs.

To derive the dual SDP, the Lagrange function can be written as
\begin{equation}
\begin{aligned}
    &\cL(x, Y_{A^n_I}, y_i, J^{\cE}_{{A^n_I A_O}}) \\
    &\coloneqq \frac{1}{p}\tr\Big[J^{\cE}_{{A^n_I A_O}}Q_{A^n_IA_O}^{T_{A^n_I}}\Big] +\Big\langle x,\tr\Big[J^{\cE}_{{A^n_I A_O}}R_{A^n_I A_O}^{T_{A^n_I}}\Big] - p\Big\rangle + \Big\langle Y_{A^n_I}, \tr_{A_O}[J^{\cE}_{{A^n_I A_O}}]-I_{A^n_I}\Big\rangle\\
    &\quad - \sum_{i}y_i\Big( \frac{1}{d_{A^n_I}}\tr\big[P_iJ^{\cE}_{A^n_IA_O}\big] - \sum_j G_{ij}x_j\Big) \\
    &=-px-\tr[Y_{A^n_I}]+ \Big\langle J^{\cE}_{{A^n_I A_O}}, \frac{1}{p}Q_{A^n_IA_O}^{T_{A^n_I}}+R_{A^n_IA_O}^{T_{A^n_I}}x+Y_{A_I^n}\ox I_{A_O} - \frac{1}{d_{A^n_I}}\sum_i y_i P_i\Big\rangle + \sum_{i,j}y_iG_{ij}x_j,
\end{aligned}
\end{equation}
where $x, Y_{A^n_I}, y_i$ are Lagrange multipliers. Then, the Lagrange dual function can be written as
\begin{align}
     \mathcal{G}(x, &Y_{A^n_I}, y_i)\coloneqq\sup_{J^{\cE}_{A^n_IA_O}\geq0}\mathcal{L}(x, Y_{A^n_I}, y_i, J^{\cE}_{{A^n_I A_O}}).
\end{align}
To make the function bounded, we have 
\begin{equation}
    \frac{1}{p}Q_{A^n_IA_O}^{T_{A^n_I}}+R_{A^n_IA_O}^{T_{A^n_I}}x+Y_{A_I^n}\ox I_{A_O} - \frac{1}{d_{A^n_I}}\sum_i y_i P_i\leq 0,
\end{equation}
$Y_{A^n_I} \leq 0$, and $\sum_i y_iG_{ij}\leq 0$.
Redefine $x$ as $x/p$, $Y_{A^n_I}$ as $Y_{A^n_I}/p$ and $y_i$ as $y_id_{A^n_I}/p$, then we obtain the following dual SDP
\begin{equation}
    \begin{aligned}
\min_{x, Y_{A^n_I}, y_i} &\,  -x - \frac{1}{p}\tr[Y_{A^n_I}]\\
 {\rm s.t.} 
&\;Q_{A^n_IA_O}^{T_{A^n_I}} + xR_{A^n_IA_O}^{T_{A^n_I}} + Y_{A^n_I}\ox I_{A_O} \leq D_{A^n_IA_O}^{T_{A^n_I}}, \\
&\;  D_{A^n_IA_O} =  \sum_iy_i(P^i_{A^n_IA_O})^{T_{A^n_I}}, \\
&\; \sum_iy_iG_{ij} \leq 0 , \; Y_{A^n_I} \leq 0,
\end{aligned}
\end{equation}
The strong duality is held by Slater's condition with $J^{\cE}_{A^n_IA_O} = pI_{A^n_IA_O}/d_{A_O}$ and $x_j = \frac{p}{|\Phi|_{n+1} }$ for $P_i = I_{A^n_IA_O}$ and $x_j=0$ otherwise, where $|\Phi|_{n+1}$ is the number of states in $\mathrm{STAB}_{n+1}$. We complete the proof. 
\end{proof}

The proof extends to the continuous state set by replacing the summation with the integral over the measure of states on $Q_{A^n_IA_O}$ and $R_{A^n_IA_O}$.


\section{Proof of Theorem~\ref{thm:nogo of CPWP puri}}\label{appendix: proof_CPWP} 

\noindent\textbf{Theorem \ref{thm:nogo of CPWP puri}} (No-go theorem for universal probabilistic purification via CPWP operations)
\textit{There is no two-to-one universal probabilistic purification protocol using CPWP operations for qudit states with odd $d$ under depolarizing noise.}

\bigskip

\begin{proof}
We denote $\lambda_0$ as the average fidelity before any purification protocol: for any pure state $\psi$ in odd-dimensional qudit systems
\begin{equation}
    \lambda_0 = \int \tr[\cD_\delta(\psi)\psi] d\psi = (1-\delta + \frac{\delta}{d}) \int d\psi = 1-\frac{d-1}{d}\delta, 
\end{equation}
where $\int d\psi = 1$. To prove Theorem \ref{thm:nogo of CPWP puri}, we will demonstrate that the maximal non-magic universal purification fidelity always equals the average fidelity before any purification protocol, i.e. 
\begin{equation}
     F^{\cA}_{\cD_\delta}(2,p) = \lambda_0,\quad \forall p\in (0,1].
\end{equation}
We will prove the theorem using the SDP in Proposition~\ref{sdp:discrete cpwp}. For two-copy inputs universal purification, the $Q_{A^2_IA_O}$ and $R_{A^2_IA_O}$  are
\begin{equation}
\begin{aligned}
   Q_{A^2_IA_O} & \coloneqq \int \cD_{\delta}(\psi)^{\ox 2} \ox \psi \,d\psi 
   = (\cD_{\delta}^{\ox 2}\ox \cI)\left(\frac{\Pi_{3}}{D(3,d)}\right),\\
   R_{A^2_IA_O} & \coloneqq \cD_{\delta}^{\ox 2}\Big(\frac{\Pi_2}{D(2,d)}\Big)\ox I,
\end{aligned}
\end{equation}
where $\Pi_{n} = D(n,d) \int \psi^{\ox n} d\psi$ is the projector onto the symmetric subspace of $\cH^{\ox n}_A$ given by the Schur's lemma ~\cite{harrow2013church, khatri2020principles, mele2024introduction}, and $D(n,d)=\tbinom{n+d-1}{n}$ is the dimension of the symmetric subspace. For odd $d$, the purifying protocol is drawn from $\cA_{\text{CPWP}}$.
We first show $\lambda_0 \geq  F^{\cA}_{\cD_\delta}(2,p)$ using the primal SDP. It is obvious to choose the following Choi operator as a feasible solution
\begin{equation}
    J^{\cE}_{A_I^2A_O} = p\tr_{A_O}\Big[J^{\cI}_{A_I^2A_O^2}\Big],
\end{equation} 
where $J^{\cI}_{A_I^2A_O^2}$ is the Choi matrix of identity channel. The objective value is then calculated as
\begin{equation}
\begin{aligned}
\frac{1}{p}\tr\Big[J^{\cE}_{A^2_IA_O}Q_{A^2_IA_O}^{T_{12}}\Big] &= \frac{1}{p}\tr\left[p(J^{\cI}_{A^2_I A_O})^{T_{12}} (\cD^{\ox 2}_\delta\ox \cI)(\Pi_{3}/D(3,d))\right]\\
&= \int \tr\left[(J^{\cI}_{A^2_I A_O})^{T_{12}}\cD_\delta(\psi)^{\ox 2}\ox \psi\right] \,d\psi\\
&=\int \tr\Big[\psi\tr_{2}\big(\cD_\delta(\psi)^{\ox 2}\big)\Big]\,d\psi\\
&= \int \tr\left[\psi\cD_\delta(\psi)\right]\,d\psi \\
&= \lambda_0,
\end{aligned}
\end{equation}
where $T_{12}$ denotes the partial transpose on the first two systems, i.e., the two input systems $A^2_I$. Similarly, we have $\tr\Big[J^{\cE}_{A^2_IA_O}R_{A^2_IA_O}^{T_{12}}\Big]= p$. Thus $p\tr_{A_O}[J^{\cI}_{A_I^2A_O^2}]$ is a feasible solution for primal SDP. To demonstrate $\lambda_0 \leq  F^{\cA}_{\cD_\delta}(2,p)$ using dual SDP, a feasible solution for the dual SDP is $x = -\lambda_0, Y_{A^2_I} = \mathbf{0}$, and 
\begin{equation}
    C_{A^2_IA_O} = 2\a I_{123} - \a[\mathbf{P}_3((123)) + \mathbf{P}_3((132))],
\end{equation}
where 
\begin{equation}
    \a =  \frac{(1 - \delta ) \delta  (d (\delta -1)-\delta )}{d^3 (d+1)},
\end{equation}
and both $\mathbf{P}_3((123)), \mathbf{P}_3((132)))$ are the permutation operators in symmetric group $\cS_3$. Notice $\alpha \leq 0$ for $\delta\in [0,1], d\geq 2$.

The following calculation of the dual program has two parts. The first part is to show that  $Q_{A^2_IA_O}^{T_{12}} - \lambda_0 R_{A^2_IA_O}^{T_{12}} - C^{T_{12}}_{A^2_IA_O} \leq 0$ since we claim $x = -\lambda_0, Y_{A^2_I} = \mathbf{0}$. The second part is to show  $c_{\mathbf{u},\mathbf{v}}\leq 0, \forall \mathbf{u,v}$ for $C_{A^2_IA_O}$. 
For the first constraint
\begin{equation}
     Q_{A^2_IA_O}^{T_{12}} - \lambda_0 R_{A^2_IA_O}^{T_{12}} - C^{T_{12}}_{A^2_IA_O} \leq 0,
\end{equation}
which is equivalent to proving
\begin{equation}
    \lambda_0 R_{A^2_IA_O}^{T_{3}} - Q_{A^2_IA_O}^{T_{3}} + C^{T_{3}}_{A^2_IA_O} \geq 0,
\end{equation}
where $T_3$ denotes the partial transpose on the third system $A_O$. 
Let $\omega = \frac{1}{t}\big(\lambda_0 R_{A^2_IA_O}^{T_{3}} - Q_{A^2_IA_O}^{T_{3}} + C^{T_{3}}_{A^2_IA_O}\big)$ where $t = \tr \Big[\lambda_0 R_{A^2_IA_O}^{T_{3}} - Q_{A^2_IA_O}^{T_{3}} + C^{T_{3}}_{A^2_IA_O}\Big]$. Notice $\omega$ is a Hermitian operator, and calculate
\begin{equation}
    t = (d-1)(1-\delta) + 2\a d(d^2 - 1) = \frac{(d-1) (1 - \delta ) \left(d^2+2 d (\delta -1) \delta -2 \delta ^2\right)}{d^2} \geq 0.
\end{equation}
We first decompose $\lambda_0 R_{A^2_IA_O}^{T_{3}} - Q_{A^2_IA_O}^{T_{3}}$ as the linear combination of partial transpose permutation operators $\mathbf{P}_3(c)^{T_3}$~\cite{yao2024protocols}
\begin{equation}
\begin{aligned}
\lambda_0 R_{A^2_IA_O}^{T_{3}} - Q_{A^2_IA_O}^{T_{3}} = & \lambda_0 \int\cD_\delta(\psi)^{\ox 2}\ox I_3\,d\psi - \int\cD_\delta(\psi)^{\ox 2}\ox\psi^{T_3}\,d\psi \\
=& \lambda_0 \int\left((1-\delta)\psi + \frac{\delta}{d}I \right)^{\ox 2}\ox I_3\,d\psi - \int\left((1-\delta)\psi + \frac{\delta}{d}I \right)^{\ox 2}\ox\psi^{T_3}\,d\psi\\
=& \lambda_0 \left((1-\delta)^2\frac{\Pi_{12}}{D(2,d)} + \frac{2\delta-\delta^2}{d^2} I_{12} \right)\ox I_3 \\
&\quad- \left((1-\delta)^3\frac{\Pi_{123}^{T_3}}{D(3,d)} + \frac{2(1-\delta)^2\delta}{d }\frac{\Pi_{12}}{D(2,d)} \ox I_3 + \frac{3\delta - 2\delta^2}{d^2}I_{123} \right)\\
= &  \left(\frac{\lambda_0d(2\delta-\delta^2)-\delta^2}{d^3}-\frac{(1-\delta)\delta}{dD(2,d)} + \frac{\lambda_0(1-\delta)^2}{2D(2,d)}-\frac{(1-\delta)^2}{6D(3,d)}  \right)I_{123}\\
&\quad + \left(\frac{\lambda_0(1-\delta)^2}{2D(2,d)}-\frac{(1-\delta)^2}{6D(3,d)}\right)\mathbf{P}_3((12))^{T_3}
 \\
&\quad - \left(\frac{(1-\delta)^2}{6D(3,d)}+\frac{(1-\delta)\delta}{2dD(2,d)}\right)\Big(\mathbf{P}_3((13))^{T_3} + \mathbf{P}_3((23))^{T_3}\Big) \\
&\quad - \frac{(1-\delta)^2}{6D(3,d)}\Big(\mathbf{P}_3((123))^{T_3} + \mathbf{P}_3((132))^{T_3}\Big),
\end{aligned}
\end{equation}
where $\Pi_{12}$ and $\Pi_{123}$ are the projectors acting on $\{1,2\}$ and $\{1,2,3\}$ system, and $D(n,d) = \tbinom{n+d-1}{n}
$. Substitute $D(2,d) = d(d+1)/2, D(3,d) =(d+2)(d+1)d/6$ and $\a$ to the  $\lambda_0 R_{A^2A}^{T_3} - Q_{A^2A}^{T_3} + C_{A^2_IA_O}^{T_3}$, 
\begin{equation}
\begin{aligned}
&\lambda_0 R_{A^2_IA_O}^{T_3} - Q_{A^2_IA_O}^{T_{3}} + C^{T_{3}}_{A^2_IA_O} \\
= & -\frac{(\delta -1) \left(\left(d^2+d-2\right) \delta ^2+(d+1) d^2-2 d \delta \right)}{d^3 (d+1) (d+2)}I_{123} -\frac{(\delta -1)^2 (d (d+1) (\delta -1)-2 \delta )}{d^2 (d+1) (d+2)}\mathbf{P}_3((12))^{T_3}\\
&\quad + \frac{(\delta -1) (d+2 \delta )}{d^2 (d+1) (d+2)} \Big(\mathbf{P}_3((13))^{T_3} + \mathbf{P}_3((23))^{T_3}\Big) \\
& \quad + \frac{(\delta -1) \left(\left(d^2+d-2\right) \delta ^2+d^2-2 (d+1) d \delta \right)}{d^3 (d+1) (d+2)}\Big(\mathbf{P}_3((123))^{T_3} + \mathbf{P}_3((132))^{T_3}\Big).
\end{aligned}
\end{equation}
Notice that $\omega$ can be represented by a basis $\{S_+,S_-,S_0,S_1,S_2,S_3\}$~\cite{eggeling2001separability}. This basis can be generated by $I\coloneqq\mathbf{P}_3((1))$, $X\coloneqq\mathbf{P}_3((23))^{T_3}$ and $V\coloneqq\mathbf{P}_3((12))$ as follows,
\begin{equation}
\begin{aligned}
&S_+ \coloneqq \frac{I+V}{2}\left(I-\frac{2X}{d+1}\right)\frac{I+V}{2},~~ S_- \coloneqq \frac{I-V}{2}\left(I-\frac{2X}{d-1}\right)\frac{I-V}{2},\\
&S_0 \coloneqq \frac{1}{d^2-1}\big[d(X+VXV)-(XV+VX)\big],\\
&S_1\coloneqq \frac{1}{d^2-1}\big[d(XV+VX)-(X+VXV)\big],\\
&S_2\coloneqq \frac{1}{\sqrt{d^2-1}}(X-VXV),~~ S_3\coloneqq\frac{i}{\sqrt{d^2-1}}(XV-VX).
\end{aligned}
\end{equation}
They satisfy $X^2=dX$, $V^2=I$, $XVX=X$, and $\tr X=d^2$, $\tr[XV]=d$. Then calculate
\begin{equation}
    \begin{aligned}s_+ \coloneqq&\tr[\omega S_+] = 
\frac{d^2 ((\delta -2) \delta +2)+d \delta ^2-2 \delta ^2}{2 d^2+4 d (\delta -1) \delta -4 \delta ^2} \\
   s_- \coloneqq&\tr[\omega S_-] = -\frac{(d-2) \delta  (d\delta -2d -\delta )}{2 d^2+4 d (\delta -1) \delta -4 \delta ^2} \\
   s_j\coloneqq&\tr[\omega S_j] = 0, \; j\in \{0,1,2,3\}.
    \end{aligned}
\end{equation}
We can check $s_+,s_-,s_0\geq 0$, $s_+ + s_- + s_0=1$ and $s_1^2+s_2^2+s_3^2\leq s_0^2$. Followed from~\cite[Lemma 2]{eggeling2001separability}, we obtain $\omega\geq 0$ which means $\lambda_0 R_{A^2A}^{T_3} - Q_{A^2A}^{T_3} + C_{A^2_IA_O}^{T_3} \geq 0$. 

Then we are about to show that the coefficients of $C_{A^2_IA_O}$ under the phase space point operator basis are negative semidefinite, i.e., $c_{\mathbf{u},\mathbf{v}}\leq 0,  \forall \mathbf{u,v}$. Notice that phase space point operators of composite system are the tensor product of phase space point operators of their subsystem system $A^\mathbf{u}_{A^2_I} = A^\mathbf{r}_{A_I} \ox A^\mathbf{w}_{A_I}$ for some $\mathbf{r}, \mathbf{w}$, and it is equivalent to express $c_{\mathbf{u},\mathbf{v}}$ as $c_{\mathbf{r}, \mathbf{w},\mathbf{v}}$. Ignore the subscripts of the systems, and calculate
\begin{equation}\label{eq:c_rwv expansion}
\begin{aligned}
    c_{\mathbf{r}, \mathbf{w},\mathbf{v}} &= \frac{1}{d^3}\tr\left[(A^{\mathbf{r}} \ox A^{\mathbf{w}} \ox A^{\mathbf{v}} )C_{A^2_IA_O}\right] \\
    &= \frac{2\a}{d^3} \tr[A^{\mathbf{r}} \ox A^{\mathbf{w}} \ox A^{\mathbf{v}} ] - \frac{\a}{d^3} \tr\Big[(A^{\mathbf{r}} \ox A^{\mathbf{w}} \ox A^{\mathbf{v}} )\big[\mathbf{P}_3((123)) + \mathbf{P}_3((132))\big]\Big] .
    \end{aligned}
\end{equation}
Using tricks of permutation operators and properties of the phase space point operator, 
we have
\begin{equation}
\begin{aligned}
\tr\big[(A^{\mathbf{r}} \ox A^{\mathbf{w}} \ox A^{\mathbf{v}} )\big] &= 1, \\
\tr\big[(A^{\mathbf{r}} \ox A^{\mathbf{w}} \ox A^{\mathbf{v}} )\mathbf{P}_3((123))\big] 
&= \tr\big[A^{\mathbf{r}}  A^{\mathbf{v}}  A^{\mathbf{w}} \big], \\
\tr\big[(A^{\mathbf{r}} \ox A^{\mathbf{w}} \ox A^{\mathbf{w}} )\mathbf{P}_3((132))\big] 
&= \tr\big[A^{\mathbf{r}}  A^{\mathbf{w}}  A^{\mathbf{v}} \big]. \\
\end{aligned}
\end{equation}
Notice that $\tr\big[A^{\mathbf{r}}  A^{\mathbf{w}}  A^{\mathbf{v}} \big]$ and $\tr\big[A^{\mathbf{r}}  A^{\mathbf{v}}  A^{\mathbf{w}} \big]$ are complex conjugate to each other.
Substitute them to~\eqref{eq:c_rwv expansion}
\begin{equation}\label{eq:c_rwv final expansion}
\begin{aligned}
    c_{\mathbf{r}, \mathbf{w},\mathbf{v}} = \frac{1}{d^3}\tr\left[(A^{\mathbf{r}} \ox A^{\mathbf{w}} \ox A^{\mathbf{v}} )C_{A^2_IA_O}\right] &= \frac{2\a}{d^3} - \frac{\a}{d^3}(\tr\left[A^{\mathbf{r}}  A^{\mathbf{w}}  A^{\mathbf{v}} \right] + \left(\tr\left[A^{\mathbf{r}}  A^{\mathbf{w}}  A^{\mathbf{v}} \right])^*\right) \\
    &= \frac{2\a}{d^3} - \frac{2\a}{d^3} \text{Re}{(\tr\left[A^{\mathbf{w}}  A^{\mathbf{v}}  A^{\mathbf{r}} \right])} \\
    &= \frac{2\a}{d^3}(1 - \text{Re}{(\tr\left[A^{\mathbf{w}}  A^{\mathbf{v}}  A^{\mathbf{r}} \right])} ),
    \end{aligned}
\end{equation}
Then we introduce
a useful property of the phase space point operator:
\begin{lemma}\label{lemma:3rd_moment_pspop}
For any three phase space point $A^{\mathbf{r}}, A^{\mathbf{w}}, A^{\mathbf{v}}$,  $\text{Re}(\tr\left[A^{\mathbf{r}} A^{\mathbf{w}} A^{\mathbf{v}}\right]) \leq 1$. 
\end{lemma}

\begin{proof} We first calculate: for any computational basis $\ket{k}$ and Heisenberg-Weyl operators $T_{\mathbf{u}}$, 
\begin{equation}
    \begin{aligned}
        T_{\mathbf{u}}\ket{k}=T_{a_1, a_2}\ket{k} &= \tau^{-a_1a_2}\omega^{a_1(k+a_2)}\ket{k \oplus a_2}. \\
    \end{aligned}
\end{equation}
where $w=e^{2\pi i/d}$ and $\oplus$ denotes addition modulo $d$. Note that $A^{\mathbf{0}} = \sum_{k\in \ZZ_d}\ket{k}\bra{-k} $~\cite{emeriau2022quantum}, then
\begin{equation}
    A^{\mathbf{u}} = T_{\mathbf{u}} A^{\mathbf{0}} T_{\mathbf{u}}^\dagger = \sum_{k\in \ZZ_d}T_{\mathbf{u}}\ket{k}\bra{-k}T_{\mathbf{u}}^\dagger = \sum_{k\in \ZZ_d}e^{\frac{2\pi i}{d}2a_1^{\mathbf{u}}k}\ket{k\oplus a_2^{\mathbf{u}}}\bra{-k \oplus a_2^{\mathbf{u}}}.
\end{equation}
Then we can calculate that
\begin{equation}
\begin{aligned}
    &\tr\left[A^{\mathbf{r}} A^{\mathbf{w}} A^{\mathbf{v}}\right] \\
    &= \sum_{l,m,n\in \ZZ_d} e^{\frac{2\pi i}{d}[2a_1^{\mathbf{r}}l + 2a_1^{\mathbf{w}}m + 2a_1^{\mathbf{v}}n]}\tr \Big[ \ket{l\oplus a_2^{\mathbf{r}}}\braket{-l \oplus a_2^{\mathbf{r}}} {m \oplus a_2^{\mathbf{w}}}\braket{-m \oplus a_2^{\mathbf{w}}}{n\oplus a_2^{\mathbf{v}}}\bra{-n \oplus a_2^{\mathbf{v}}} \Big] \\
    &= \sum_{l,m,n\in \ZZ_d} e^{\frac{4\pi i}{d}[a_1^{\mathbf{r}}l + a_1^{\mathbf{w}}m + a_1^{\mathbf{v}}n]} \delta(-l \oplus a_2^{\mathbf{r}} , m \oplus a_2^{\mathbf{w}})\delta(-m \oplus a_2^{\mathbf{w}} , n \oplus a_2^{\mathbf{v}})\delta(-n \oplus a_2^{\mathbf{v}}, l\oplus a_2^{\mathbf{r}}).
\end{aligned}
\end{equation}
To further simplify the delta function, solve the system of equations
\begin{equation}
\begin{cases}
-l\oplus a_2^{\mathbf{r}} = m \oplus a_2^{\mathbf{w}} \\
-m \oplus a_2^{\mathbf{w}} = n \oplus a_2^{\mathbf{v}} \\
-n \oplus a_2^{\mathbf{v}} = l\oplus a_2^{\mathbf{r}}
\end{cases}
\end{equation}
in terms of the variables $l,m,n$ under the finite field $\ZZ_d$. The solution is
\begin{equation}
\begin{cases}
l = -a_2^{\mathbf{w}} \oplus a_2^{\mathbf{v}} \\
m = -a_2^{\mathbf{v}} \oplus a_2^{\mathbf{r}} \\
n = -a_2^{\mathbf{r}} \oplus a_2^{\mathbf{w}}.
\end{cases}
\end{equation}
Substitute them back to obtain
\begin{equation}
\begin{aligned}
    &\tr\left[A^{\mathbf{r}} A^{\mathbf{w}} A^{\mathbf{v}}\right] \\
    &= \sum_{l,m,n}e^{\frac{4\pi i}{d}[a_1^{\mathbf{r}}l + a_1^{\mathbf{w}}m + a_1^{\mathbf{v}}n]} \delta( l, -a_2^{\mathbf{w}} \oplus a_2^{\mathbf{v}})\delta(m, -a_2^{\mathbf{v}} \oplus a_2^{\mathbf{r}})\delta(n, -a_2^{\mathbf{r}} \oplus a_2^{\mathbf{w}}) \\
    &= e^{\frac{4 \pi i}{d}[a_1^{\mathbf{r}}(a_2^{\mathbf{v}}-a_2^{\mathbf{w}}) + a_1^{\mathbf{w}}(a_2^{\mathbf{r}}-a_2^{\mathbf{v}}) + a_1^{\mathbf{v}}(a_2^{\mathbf{w}}-a_2^{\mathbf{r}})]} \\
    &= e^{\frac{4 \pi i}{d}f},
\end{aligned}
\end{equation}
where we replace the modular addition $\oplus$ with the common addition due to periodicity of the phase, and denote the expression $f = a_1^{\mathbf{r}}(a_2^{\mathbf{v}}-a_2^{\mathbf{w}}) + a_1^{\mathbf{w}}(a_2^{\mathbf{r}}-a_2^{\mathbf{v}}) + a_1^{\mathbf{v}}(a_2^{\mathbf{w}}-a_2^{\mathbf{r}})$. Thus we conclude that $\text{Re}(\tr\left[A^{\mathbf{r}} A^{\mathbf{w}} A^{\mathbf{v}}\right]) = \cos(\frac{4\pi f}{d}) \leq 1$.
\end{proof}

Using Lemma~\ref{lemma:3rd_moment_pspop}, and also observing $\a\leq 0$, it follows that $c_{\mathbf{r}, \mathbf{w},\mathbf{v}} \leq 0,  \forall \mathbf{r,w,v}$ for Eq.~\eqref{eq:c_rwv final expansion}. We complete the calculation for the dual SDP. Combining $\lambda_0 \geq  F^{\cA}_{\cD_\delta}(2,p)$ from the primal SDP and $\lambda_0 \leq  F^{\cA}_{\cD_\delta}(2,p)$ from the dual SDP, we conclude that $\lambda_0 =  F^{\cA}_{\cD_\delta}(2,p)$ and the proof is complete. 
\end{proof}


\section{Proof of Theorem~\ref{thm:nogo of CSPO puri}}\label{appendix: proof_CSPO} 

\noindent \textbf{Theorem \ref{thm:nogo of CSPO puri}} (No-go theorem for universal probabilistic purification via CSPOs)
\textit{There is no two-to-one universal probabilistic purification protocol using CSPOs for noisy qubit states under depolarizing noise.}

\bigskip

\begin{proof}
We still denote $\lambda_0$ as the average fidelity before any purification protocol: for any pure state $\psi$ in a qubit system
\begin{equation}
    \lambda_0 = \int \tr[\cD_\delta(\psi)\psi] d\psi  = 1-\frac{1}{2}\delta.
\end{equation}
We will still show that the maximal universal fidelity with $\cA$ for two copies is
\begin{equation}
     F^{\cA}_{\cD_\delta}(2,p) = \lambda_0,\quad \forall p\in (0,1].
\end{equation}
For qubit, the purifying protocol is drawn from $\cA_{\text{CSPO}}$. The proof is similar to $\cA_{\text{CPWP}}$ using the SDP in Proposition~\ref{sdp:discrete cspo}. The feasible solution of the primal SDP is also $J^{\cE}_{A_I^2A_O} = p\tr_{A_O}[J^{\cI}_{A_I^2A_O^2}]$ where $J^{\cI}_{A_I^2A_O^2}$ is the Choi matrix of identity channel,
whose Choi state is a stabilizer state,  satisfying $\sum_j G_{ij}x_j = b_i, x_j \geq 0, \; \forall j$. The feasible solution of dual SDP is identical to $\cA_{\text{CPWP}}$ with $d=2$: $x = -\lambda_0, Y_{A^2_I} = \mathbf{0}$, and 
\begin{equation}
    D_{A^2_IA_O} = 2\beta I_{123} - \beta(\mathbf{P}_3((123)) + \mathbf{P}_3((132)))
\end{equation}
where \begin{equation}
    \beta = \alpha|_{d=2} = \frac{(1 - \delta ) \delta  (\delta - 2 )}{24}.
\end{equation}

As we have shown that  $Q_{A^2_IA_O}^{T_{12}} - \lambda_0 R_{A^2_IA_O}^{T_{12}} - D^{T_{12}}_{A^2_IA_O} \leq 0$ which is the same as the $\cA_{\text{CPWP}}$ case, 
we just need to calculate  $\sum_iy_iG_{ij} \leq 0$ for $D_{A^2_IA_O}$, where $G_{ij} = \tr [P_i \ketbra{\phi_j}{\phi_j}], \ketbra{\phi_j}{\phi_j} \in \mathrm{STAB}_{3}$, and $\{P_i\}$ are Pauli operators for three qubit system. 

For two copies, $D_{A^2_IA_O}$ is written as 
\begin{equation}
    D_{A^2_IA_O} = \sum_i y_i(P^i_{A^2_IA_O})^{T_{12}} = \sum_{r,w,v=0}^3y_{r,w,v}(\sigma^r_{A_I} \ox \sigma^w_{A_I})^T \ox \sigma^v_{A_O},
\end{equation}
where $\sigma^i$ denote the single qubit Pauli matrices $\sigma^i \in \{I, X, Y, Z\}$ and $i \in \{0,1,2,3\}$. Ignore the subscripts of the systems, and calculate each coefficient $y_{r,w,v}$, 
\begin{equation*}
\begin{aligned}
   y_{r,w,v} &= \frac{1}{8}\tr\big[(\sigma ^r\ox \sigma^w  \ox \sigma^v )D_{A^2_IA_O}\big] \\
    &= \frac{(-1)^{m_r + m_w}}{8} \bigg[ 2\beta\tr [\sigma ^r\ox \sigma^w \ox \sigma^v ] -  \beta  \tr\Big[(\sigma ^r\ox \sigma^w \ox \sigma^v )[\mathbf{P}_3((123)) + \mathbf{P}_3((132))]\Big] \bigg],
\end{aligned}
\end{equation*}
where $m_r,m_w$ are sign factors due to the transpose of the Pauli $Y$ operator. When $\sigma^r=Y$, $m_r = 1$; otherwise $m_r=0$, and so is $m_w$. 
Still using the tricks of permutations 
\begin{align}
    \tr \big[\sigma ^r\ox \sigma^w \ox \sigma^v \big] &= 8\delta(r, 0)\delta(w, 0)\delta(v, 0), \\
    \tr\big[(\sigma ^r\ox \sigma^w \ox \sigma^v )\mathbf{P}_3((123)\big] &= \tr\big[\sigma ^r \sigma^v   \sigma^w \big], \\
     \tr\big[(\sigma ^r\ox \sigma^w \ox \sigma^v )\mathbf{P}_3((132)\big] &= \tr\big[\sigma ^r \sigma^w  \sigma^v \big].
\end{align}
$\tr\big[\sigma ^r \sigma^v   \sigma^w \big]$ and $\tr\big[\sigma ^r \sigma^w  \sigma^v \big]$ are complex conjugate to each other. We introduce 
a property for Pauli operators: for any three single qubit Pauli operators $\sigma ^l, \sigma ^m, \sigma ^n$,
\begin{equation}
    \tr\big[\sigma ^l \sigma ^m \sigma ^n\big] = 2i\epsilon_{0lmn} + 2\left[\delta(l,0)\delta(m,n) + \delta(n,0)\delta(l,m) + \delta(m,0)\delta(n,l)\right] - 4\delta(l,0)\delta(m,0)\delta(n,0),
\end{equation}
where $\epsilon_{0lmn}$ is the Levi-Civita symbol. Substitute them with
\begin{equation}
    \begin{aligned}
     y_{r,w,v} &= \frac{(-1)^{m_r + m_w}}{8}\Big[16\beta\delta(r, 0)\delta(w, 0)\delta(v, 0) -  2\beta\cdot \text{Re} \left( \tr[\sigma ^r \sigma^v   \sigma^w ] \right)\Big] \\
     &= \frac{(-1)^{m_r + m_w}}{8} \bigg[ 16\beta \delta(r, 0)\delta(w, 0)\delta(v, 0) - 2\beta \Big[ 2\big(\delta(r,0)\delta(w,v) + \delta(w,0)\delta(r,v)\\ 
     &\quad + \delta(v,0)\delta(r,w)\big) - 4\delta(r,0)\delta(w,0)\delta(v,0) \Big] \bigg] \\
     &= 3\beta \delta(r, 0)\delta
     (w, 0)\delta(v, 0) - \frac{\beta}{2}\big[\delta(r,w)\delta(v, 0) +  (-1)^{ m_w}\delta(r,0)\delta(w,v) + (-1)^{m_r}\delta(w,0)\delta(r,v)\big].
    \end{aligned}
\end{equation}
For a single qubit system with two-copy inputs, we have
\begin{equation}
    G_{ij} = \tr [P_i \ketbra{\phi_j}{\phi_j}] = \tr\big[(\sigma ^r\ox \sigma^w  \ox \sigma^v )\ketbra{\phi_j}{\phi_j}\big]
\end{equation}
where $\ketbra{\phi_j}{\phi_j} \in \mathrm{STAB}_{3}$. Then calculate $\sum_iy_iG_{ij}$, 
\begin{equation}
\begin{aligned}
\sum_iy_iG_{ij} &= \sum_{r,w,v=0}^3 y_{r,w,v}\tr\big[(\sigma ^r\ox \sigma^w  \ox \sigma^v )\ketbra{\phi_j}{\phi_j}\big] \\
& = \sum_{r,w,v=0}^3 \Big[ 3\beta \delta(r, 0)\delta(w, 0)\delta(v, 0) - \frac{\beta}{2}\Big( \delta(r,w)\delta(v, 0)\\
&\qquad +  (-1)^{m_w}\delta(r,0)\delta(w,v) + (-1)^{m_r}\delta(w,0)\delta(r,v)\Big) \Big]\tr\big[(\sigma ^r\ox \sigma^w  \ox \sigma^v )\ketbra{\phi_j}{\phi_j}\big] \\
&= 3\beta - \frac{\beta}{2} \Big[ \sum_{r=0}^3  \tr\big[(\sigma ^r\ox \sigma^r  \ox I)\ketbra{\phi_j}{\phi_j}\big] + \sum_{w=0}^3 (-1)^{m_w} \tr\big[(I \ox \sigma^w  \ox \sigma^w )\ketbra{\phi_j}{\phi_j}\big] \\
&\qquad + \sum_{r=0}^3 (-1)^{m_r} \tr\big[(\sigma ^r\ox I \ox \sigma^r )\ketbra{\phi_j}{\phi_j}\big] \Big]\\
&= \frac{3}{2}\beta - \frac{\beta}{2}  \sum_{r=1}^3  \Big[ \tr\big[(\sigma ^r\ox \sigma^r  \ox I)\ketbra{\phi_j}{\phi_j}\big] + (-1)^{m_r}\tr\big[(I \ox \sigma^r  \ox \sigma^r )\ketbra{\phi_j}{\phi_j}\big]\\
&\qquad + (-1)^{m_r}\tr\big[(\sigma ^r\ox I \ox \sigma^r )\ketbra{\phi_j}{\phi_j}\big] \Big].
\end{aligned}
\end{equation}
It is straightforward to check that $\sum_iy_iG_{ij}$ take the values $3\beta, 2\beta, \frac{3}{2}\beta, \beta, 0$ depending on the stabilizer state $\ket{\phi_j}$. Since $\beta\leq 0$, we conclude that $\sum_iy_iG_{ij} \leq 0, \forall j$. We complete the calculation of the dual SDP solution for $\cA_{\text{CSPO}}$. Combining $\lambda_0 \geq  F^{\cA}_{\cD_\delta}(2,p)$ from the primal SDP and $\lambda_0 \leq  F^{\cA}_{\cD_\delta}(2,p)$ from the dual SDP, we conclude that $\lambda_0 =  F^{\cA}_{\cD_\delta}(2,p)$ and the proof is complete. 
\end{proof}

\end{document}

%% file: pretex.tex
\usepackage[dvipsnames]{xcolor}
\usepackage{framed}
\definecolor{shadecolor}{rgb}{0.9,0.9,0.9}

\usepackage{mathtools}
\usepackage{amsmath}
\usepackage[shortlabels]{enumitem}

\usepackage{graphicx,epic,eepic,epsfig,amsmath,latexsym,amssymb,verbatim,color}
 
\usepackage{amsfonts}       
\usepackage{nicefrac}       

\usepackage{amsmath}
\usepackage{bbm}

\usepackage{float}
\usepackage{tikz}
\usetikzlibrary{chains}
\usetikzlibrary{fit}
\usepackage{pgflibraryarrows}		
\usepackage{pgflibrarysnakes}		

\usepackage{epsfig}
\usetikzlibrary{shapes.symbols,patterns} 
\usepackage{pgfplots}

\usepackage[strict]{changepage}
\usepackage{hyperref}
\hypersetup{colorlinks=true,citecolor=blue,linkcolor=blue,filecolor=blue,urlcolor=blue,breaklinks=true}

\usepackage[marginal]{footmisc}
\usepackage{url}
\usepackage{theorem}

\newtheorem{definition}{Definition}
\newtheorem{proposition}{Proposition}
\newtheorem{lemma}[proposition]{Lemma}

\newtheorem{theorem}[proposition]{Theorem}

\def\squareforqed{\hbox{\rlap{$\sqcap$}$\sqcup$}}
\def\qed{\ifmmode\squareforqed\else{\unskip\nobreak\hfil
\penalty50\hskip1em\null\nobreak\hfil\squareforqed
\parfillskip=0pt\finalhyphendemerits=0\endgraf}\fi}
\def\endenv{\ifmmode\;\else{\unskip\nobreak\hfil
\penalty50\hskip1em\null\nobreak\hfil\;
\parfillskip=0pt\finalhyphendemerits=0\endgraf}\fi}
\newenvironment{proof}{\noindent \textbf{{Proof~} }}{\hfill $\blacksquare$}

\newcounter{remark}

\newcounter{example}

\mathchardef\ordinarycolon\mathcode`\:
\mathcode`\:=\string"8000
\def\vcentcolon{\mathrel{\mathop\ordinarycolon}}
\begingroup \catcode`\:=\active
  \lowercase{\endgroup
  \let :\vcentcolon
  }

\usepackage{cleveref}
\usepackage{graphicx}
\usepackage{xcolor}

\RequirePackage[framemethod=default]{mdframed}
\newmdenv[skipabove=7pt,
skipbelow=7pt,
backgroundcolor=darkblue!15,
innerleftmargin=5pt,
innerrightmargin=5pt,
innertopmargin=5pt,
leftmargin=0cm,
rightmargin=0cm,
innerbottommargin=5pt,
linewidth=1pt]{tBox}

\newmdenv[skipabove=7pt,
skipbelow=7pt,
backgroundcolor=red!15,
innerleftmargin=5pt,
innerrightmargin=5pt,
innertopmargin=5pt,
leftmargin=0cm,
rightmargin=0cm,
innerbottommargin=5pt,
linewidth=1pt]{rBox}

\newmdenv[skipabove=7pt,
skipbelow=7pt,
backgroundcolor=blue2!25,
innerleftmargin=5pt,
innerrightmargin=5pt,
innertopmargin=5pt,
leftmargin=0cm,
rightmargin=0cm,
innerbottommargin=5pt,
linewidth=1pt]{dBox}
\newmdenv[skipabove=7pt,
skipbelow=7pt,
backgroundcolor=darkkblue!15,
innerleftmargin=5pt,
innerrightmargin=5pt,
innertopmargin=5pt,
leftmargin=0cm,
rightmargin=0cm,
innerbottommargin=5pt,
linewidth=1pt]{sBox}
\definecolor{darkblue}{RGB}{0,76,156}
\definecolor{darkkblue}{RGB}{0,0,153}
\definecolor{blue2}{RGB}{102,178,255}
\definecolor{darkred}{RGB}{195,0,0}

\newcommand{\nc}{\newcommand}
\nc{\rnc}{\renewcommand}
\nc{\lbar}[1]{\overline{#1}}
\nc{\bra}[1]{\langle#1|}
\nc{\ket}[1]{|#1\rangle}
\nc{\ketbra}[2]{|#1\rangle\!\langle#2|}
\nc{\braket}[2]{\langle#1|#2\rangle}

\nc{\proj}[1]{| #1\rangle\!\langle #1 |}
\nc{\avg}[1]{\langle#1\rangle}
\nc{\smfrac}[2]{\mbox{$\frac{#1}{#2}$}}
\nc{\tr}{\operatorname{Tr}}
\nc{\ox}{\otimes}
\nc{\dg}{\dagger}
\nc{\dn}{\downarrow}
\nc{\cA}{{\cal A}}
\nc{\cB}{{\cal B}}
\nc{\cC}{{\cal C}}
\nc{\cD}{{\cal D}}
\nc{\cE}{{\cal E}}
\nc{\cF}{{\cal F}}
\nc{\cG}{{\cal G}}
\nc{\cH}{{\cal H}}
\nc{\cI}{{\cal I}}
\nc{\cJ}{{\cal J}}
\nc{\cK}{{\cal K}}
\nc{\cL}{{\cal L}}
\nc{\cM}{{\cal M}}
\nc{\cN}{{\cal N}}
\nc{\cO}{{\cal O}}
\nc{\cP}{{\cal P}}
\nc{\cQ}{{\cal Q}}
\nc{\cR}{{\cal R}}
\nc{\cS}{{\cal S}}
\nc{\cT}{{\cal T}}
\nc{\cU}{{\cal U}}
\nc{\cV}{{\cal V}}
\nc{\cX}{{\cal X}}
\nc{\cY}{{\cal Y}}
\nc{\cZ}{{\cal Z}}
\nc{\cW}{{\cal W}}
\nc{\csupp}{{\operatorname{csupp}}}
\nc{\qsupp}{{\operatorname{qsupp}}}
\nc{\var}{{\operatorname{var}}}
\nc{\rar}{\rightarrow}
\nc{\lrar}{\longrightarrow}
\nc{\polylog}{{\operatorname{polylog}}}
\nc{\wt}{{\operatorname{wt}}}
\nc{\av}[1]{{\left\langle {#1} \right\rangle}}
\nc{\supp}{{\operatorname{supp}}}

\nc{\argmin}{{\operatorname{argmin}}}

\def\a{\alpha}

\def\x{\xi}

\nc{\RR}{{{\mathbb R}}}
\nc{\CC}{{{\mathbb C}}}
\nc{\FF}{{{\mathbb F}}}
\nc{\NN}{{{\mathbb N}}}
\nc{\ZZ}{{{\mathbb Z}}}
\nc{\PP}{{{\mathbb P}}}
\nc{\QQ}{{{\mathbb Q}}}
\nc{\UU}{{{\mathbb U}}}
\nc{\EE}{{{\mathbb E}}}
\nc{\id}{{\operatorname{id}}}

\nc{\CHSH}{{\operatorname{CHSH}}}

\nc{\be}{\begin{equation}}
\nc{\ee}{{\end{equation}}}
\nc{\bea}{\begin{eqnarray}}
\nc{\eea}{\end{eqnarray}}
\nc{\<}{\langle}
\rnc{\>}{\rangle}
\nc{\rU}{\mbox{U}}

\nc{\ob}[1]{#1}

\nc{\SEP}{{\text{\rm SEP}}}
\nc{\NS}{{\text{\rm NS}}}
\nc{\LOCC}{{\text{\rm LOCC}}}
\nc{\PPT}{{\text{\rm PPT}}}
\nc{\EXT}{{\text{\rm EXT}}}
\nc{\Sym}{{\operatorname{Sym}}}

\nc{\ERLO}{{E_{\text{r,LO}}}}
\nc{\ERLOCC}{{E_{\text{r,LOCC}}}}
\nc{\ERPPT}{{E_{\text{r,PPT}}}}
\nc{\ERLOCCinfty}{{E^{\infty}_{\text{r,LOCC}}}}
\nc{\Aram}{{\operatorname{\sf A}}}

\usepackage{tikz}
\usepackage{hyperref}
\hypersetup{colorlinks=true,citecolor=blue,linkcolor=blue,filecolor=blue,urlcolor=blue,breaklinks=true}

\makeatletter
\def\grd@save@target#1{%
  \def\grd@target{#1}}
\def\grd@save@start#1{%
  \def\grd@start{#1}}
\tikzset{
  grid with coordinates/.style={
    to path={%
      \pgfextra{%
        \edef\grd@@target{(\tikztotarget)}%
        \tikz@scan@one@point\grd@save@target\grd@@target\relax
        \edef\grd@@start{(\tikztostart)}%
        \tikz@scan@one@point\grd@save@start\grd@@start\relax
        \draw[minor help lines,magenta] (\tikztostart) grid (\tikztotarget);
        \draw[major help lines] (\tikztostart) grid (\tikztotarget);
        \grd@start
        \pgfmathsetmacro{\grd@xa}{\the\pgf@x/1cm}
        \pgfmathsetmacro{\grd@ya}{\the\pgf@y/1cm}
        \grd@target
        \pgfmathsetmacro{\grd@xb}{\the\pgf@x/1cm}
        \pgfmathsetmacro{\grd@yb}{\the\pgf@y/1cm}
        \pgfmathsetmacro{\grd@xc}{\grd@xa + \pgfkeysvalueof{/tikz/grid with coordinates/major step}}
        \pgfmathsetmacro{\grd@yc}{\grd@ya + \pgfkeysvalueof{/tikz/grid with coordinates/major step}}
        \foreach \x in {\grd@xa,\grd@xc,...,\grd@xb}
        \node[anchor=north] at (\x,\grd@ya) {\pgfmathprintnumber{\x}};
        \foreach \y in {\grd@ya,\grd@yc,...,\grd@yb}
        \node[anchor=east] at (\grd@xa,\y) {\pgfmathprintnumber{\y}};
      }
    }
  },
  minor help lines/.style={
    help lines,
    step=\pgfkeysvalueof{/tikz/grid with coordinates/minor step}
  },
  major help lines/.style={
    help lines,
    line width=\pgfkeysvalueof{/tikz/grid with coordinates/major line width},
    step=\pgfkeysvalueof{/tikz/grid with coordinates/major step}
  },
  grid with coordinates/.cd,
  minor step/.initial=.2,
  major step/.initial=1,
  major line width/.initial=2pt,
}
\makeatother

\usepackage{thmtools}
\usepackage{thm-restate}
\usepackage{etoolbox}
\makeatletter
\def\problem@s{}
\newcounter{problems@cnt}

\newcommand{\allproblems}{\problem@s}
\makeatother